\documentclass[11pt]{article}

\usepackage{amsmath,amssymb}
\usepackage{amsthm} 
\usepackage{newtxtext}
\usepackage{newtxmath}
\usepackage[T1]{fontenc}
\usepackage[utf8]{inputenc}
\usepackage[margin=1in]{geometry}
\usepackage{xcolor}
\usepackage{graphicx}
\usepackage[expansion=false]{microtype}
\usepackage[colorlinks=true,linkcolor=blue!50!black,citecolor=blue!50!black]{hyperref}

\theoremstyle{plain}
\newtheorem{thm}{Theorem}[section]
\newtheorem{prop}[thm]{Proposition}
\newtheorem{lem}[thm]{Lemma}
\newtheorem{cor}[thm]{Corollary}
\theoremstyle{definition}
\newtheorem{defn}[thm]{Definition}
\newtheorem{ex}[thm]{Example}
\theoremstyle{remark}
\newtheorem{rem}[thm]{Remark}

\newcommand{\M}{\mathcal{M}}
\newcommand{\R}{\mathbb{R}}
\newcommand{\Z}{\mathbb{Z}}
\newcommand{\C}{\mathbb{C}}
\newcommand{\PL}{\textup{P\L}}
\newcommand{\KL}{\textup{K\L}}
\DeclareMathOperator{\dist}{dist}
\DeclareMathOperator{\Id}{Id}

\newcommand{\Log}{\operatorname{Log}}
\newcommand{\Exp}{\operatorname{Exp}}
\newcommand{\bcr}{\textup{BCR}}

\title{Smooth globally P\L{} functions are nonlinear least-squares,\\
and so are their gradient-dominated cousins}
\author{Eduardo Sontag\\
  Northeastern University\\
  \mbox{\tt e.sontag@northeastern.edu, sontag@sontaglab.org}\\
  \mbox{\tt https://sontaglab.org}}
\date{}

\begin{document}

\maketitle

\begin{abstract}
  \noindent
Boumal, Criscitiello and Rebjock (BCR) proved that if $\M$ is a
contractible, connected and complete Riemannian manifold, then every
smooth function $f\colon\M\to\R$ satisfying the global
Polyak--\L{}ojasiewicz inequality (P\L{}I) is necessarily of the form
$f = f^* + \|\phi\|^2$ with $\phi$ a submersion.  Informally,
minimizing such a function amounts to solving a nonlinear
least-squares problem in new coordinates.  The global P\L{}I
hypothesis fails, however, in many problems of interest, among them
continuous-time LQR policy optimization in optimal control and a
standard formulation of logistic regression.  A hierarchy of weakened
P\L{} inequalities has been introduced in order to cover such
problems, and more generally to study the effect of noise and
adversarial perturbations on gradient flows.  This note shows that,
with minor modifications, the same reduction to a nonlinear
least-squares problem holds under a substantially weaker hypothesis,
``semiglobal'' P\L{}I, which is satisfied in both of the examples just
mentioned.  That condition asks that $f$ satisfy an estimate
$\|\nabla f(x)\| \ge \alpha\bigl(f(x)-f^*\bigr)$ for all $x$, with
$\alpha$ merely positive definite and bounded below by a positive
multiple of $\sqrt{s}$ for small $s>0$.

\end{abstract}

\tableofcontents

\section{Introduction}

\subsection{Why weaken the P\L{} inequality?}
\label{sec:why}

The Polyak--\L{}ojasiewicz inequality is invoked in optimization for two
reasons. The first, and the original one \cite{Polyak1963}, is that it delivers a
global exponential rate: along the gradient flow
$\dot k = -\nabla\mathcal{L}(k)^\top$, writing
$\ell(t) = \mathcal{L}(k(t)) - \underline{\mathcal{L}}$, where
$\underline{\mathcal{L}} := \inf_{D}\mathcal{L}$ denotes the minimal value of the
loss, one has
$\dot\ell = -\|\nabla\mathcal{L}\|^2 \le -\lambda\ell$ and hence
$\ell(t)\le e^{-\lambda t}\ell(0)$. The second, more recent, is robustness: if
the gradient is supplied by an ``oracle'' (a simulator, a digital twin, an
experiment on hardware, a sample from limited data) then what one actually
integrates is $\dot k = -\nabla\mathcal{L}(k)^\top + d$, and one wants the
graceful degradation encoded by input-to-state stability
\cite{Sontag2022, CuiJiangSontag2024, Sontag2007}.

The difficulty is that in several problems where one most wants these
conclusions, the global inequality simply fails, and it typically
fails because the gradient stays bounded while the loss does not.
Two examples, both from \cite{Sontag2025L4DC}, make the point.

\smallskip
\emph{Continuous-time LQR.} Take the scalar integrator $\dot x = u$, $x(0)=1$,
with cost $\int_0^\infty x^2 + u^2$ and $u = -kx$. Then
\begin{equation}\label{eq:lqrscalar}
  \mathcal{L}(k) \ =\ \frac{1+k^2}{2k} ,
  \qquad k \in D = (0,+\infty) \quad [k \text{ stabilizing}] ,
\end{equation}
with $\underline{\mathcal{L}} = 1$ at $k=1$. Here
$|\nabla\mathcal{L}(k)|^2 = \tfrac14(1-k^{-2})^2 \to \tfrac14$ as $k\to\infty$,
while $\mathcal{L}(k)-\underline{\mathcal{L}}\to\infty$. So there is no estimate
$\|\nabla\mathcal{L}\| \ge \alpha(\mathcal{L}-\underline{\mathcal{L}})$ with
$\alpha$ unbounded, and in particular none with $\alpha(r) = \sqrt{\lambda r}$.
What one sees instead is an almost linear decrease at rate $\approx\tfrac14$,
followed by a ``soft switch'' to exponential decay near the minimum
(Figure~\ref{fig:lqr}). This is in sharp contrast with the \emph{discrete-time}
LQR problem, which does satisfy the global inequality \cite{Fazel2018}.

\smallskip
\emph{Logistic regression.} Take two samples $(x_1,y_1) = (1,0)$ and
$(x_2,y_2) = (2,1)$ with scalar parameter $k$. The binary cross-entropy loss is
\begin{equation}\label{eq:logistic}
  \mathcal{L}(k) \ =\ \log(1+e^{k}) + \log(1+e^{-2k}) ,
  \qquad
  \mathcal{L}'(k) \ =\ \frac{1}{1+e^{-k}} - \frac{2}{1+e^{2k}} .
\end{equation}
Again $\mathcal{L}(k)\to+\infty$ as $k\to\pm\infty$, while
$\mathcal{L}'(k) \to 1$ and $\to -2$ respectively: bounded gradient, unbounded
loss, no unbounded $\alpha$.

\smallskip
These obstructions are what motivated the study, in \cite{Sontag2025L4DC} and the
work summarized there, of gradient dominance estimates
$\|\nabla\mathcal{L}(k)\| \ge \alpha(\mathcal{L}(k)-\underline{\mathcal{L}})$ for
general \emph{comparison functions} $\alpha$, organized into the hierarchy
$gl$-P\L{}I $\Rightarrow$ $sat$-P\L{}I $\Rightarrow$ $sgl$-P\L{}I $\Rightarrow$
$loc$-P\L{}I recalled in Section~\ref{sec:hierarchy}. The class singled out
there as ``saturated'',
\[
  sat\text{-P\L{}I}: \qquad \alpha(r) \ =\ \sqrt{\frac{ar}{b+r}}
  \qquad (a,b>0) ,
\]
interpolates exactly between the two regimes visible in Figure~\ref{fig:lqr}:
$\alpha(r) \approx \sqrt{(a/b)\,r}$ near the origin, giving local exponential
convergence, and $\alpha(r)\approx\sqrt a$ for large $r$, giving global linear
convergence. Continuous-time LQR is asserted to satisfy it in
\cite{Sontag2025L4DC}, and satisfies at least the class-$\mathcal{K}$ version
proved in \cite{CuiJiangSontag2024}; logistic regression satisfies the
class-$\mathcal{K}$ version too, called the $\mathcal{K}$-P\L{} condition in
\cite{CuiJiangSontag2026}, where it is used to establish noise-to-state stability
of the corresponding stochastic gradient dynamics. Section~\ref{sec:applications}
sorts out which part of each estimate does which job.

Now, \cite{BCR2026} assumes the global inequality \eqref{eq:PL} throughout.
Neither of the two examples above satisfies it. It is therefore natural to ask
whether the structure theory of \cite{BCR2026}, which is about the shape of the
function and of its minimizer set, not about rates, is really tied to that
hypothesis. This note argues that it is not: everything survives under
$sgl$-P\L{}I, hence under $sat$-P\L{}I, hence in both examples above. Section
\ref{sec:applications} works out what the conclusions say there.

\begin{figure}[htbp]
\centering
\IfFileExists{lqr-scalar.pdf}{\includegraphics[width=\textwidth]{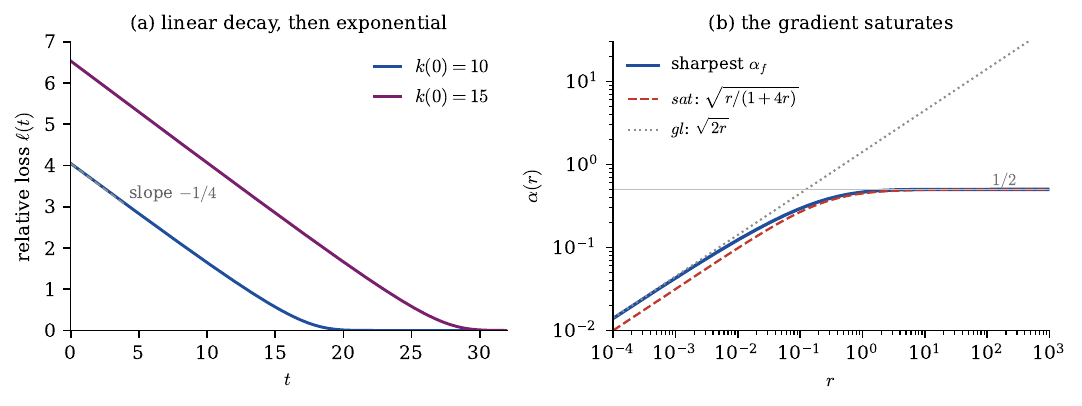}}{%
\fbox{\parbox{0.9\textwidth}{\centering\vspace{2em}\itshape
[Figure \texttt{lqr-scalar.pdf} not found; see the companion script.]
\vspace{2em}}}}
\caption{The scalar continuous-time LQR loss \eqref{eq:lqrscalar}.
\emph{(a)} Relative loss along the gradient flow from $k(0)=10$ and $k(0)=15$:
an almost linear decrease at rate $\approx\tfrac14$, then a soft switch to
exponential decay. \emph{(b)} The sharpest comparison function
$\alpha_f(r) = \inf\{|\nabla\mathcal{L}| : \mathcal{L}-\underline{\mathcal{L}}=r\}$,
which behaves like $\sqrt{2r}$ at the origin but saturates at $\tfrac12$; the
$sat$-P\L{}I bound $\sqrt{r/(1+4r)}$ (that is, $a=b=\tfrac14$) is valid and
asymptotically tight, whereas no $gl$-P\L{}I bound $\sqrt{\lambda r}$ can hold.}
\label{fig:lqr}
\end{figure}

\subsection{Setting, and a change of notation}
\label{sec:setting}

From here until Section~\ref{sec:applications} we adopt the notation of
\cite{BCR2026}, so that our statements can be compared with theirs line by line.
The dictionary is
\[
  \mathcal{L} \rightsquigarrow f , \qquad
  \underline{\mathcal{L}} \rightsquigarrow f^* , \qquad
  k \rightsquigarrow x , \qquad
  D \rightsquigarrow \M , \qquad
  \mathcal{T}_{\mathcal{L}} \rightsquigarrow S .
\]
Section~\ref{sec:applications} switches back, because there the letter $x$ is
needed for the state of a control system and $k$ for the feedback matrix being
optimized.

Let $\M$ be a smooth, connected, complete Riemannian manifold without boundary,
of dimension $n$, and let $f\colon\M\to\R$ be smooth ($C^\infty$) and bounded
below, with $f^* := \inf_{\M} f$. Write
\[
  h(x) := f(x) - f^* \ \ge 0 .
\]
Recall that $f$ is \emph{globally $\mu$-P\L{}} if
\begin{equation}
  \|\nabla f(x)\|^2 \ \ge\ 2\mu\, h(x) \qquad \text{for all } x \in \M .
  \tag{\PL}\label{eq:PL}
\end{equation}
The paper \cite{BCR2026} extracts a remarkable amount of global structure from
this pointwise, first-order inequality. Its main result is that if $\M$ is
contractible then $f = f^* + \|\phi\|^2$ for a smooth submersion
$\phi\colon\M\to\R^k$, with $k$ the codimension of the minimizer set $S$; and
its corollaries pin down exactly which manifolds can arise as $S$, when $f$ can
be straightened into a convex quadratic, and so on.

Our purpose here is to ask what \eqref{eq:PL} is actually being used for. It
seems to us that it does three separable jobs.

\begin{enumerate}
\item[(J1)] \emph{Near $S$}, it forces the \L{}ojasiewicz exponent to be
$\tfrac12$. This is what makes $S$ a manifold at all, and what makes the Hessian
of $f$ uniformly positive definite in the normal directions. It is the reason
Morse--Bott theory becomes available.
\item[(J2)] \emph{Away from $S$}, it forces the gradient to be bounded away from
zero on each level set. This is a uniform strengthening of invexity, and it is
what makes negative gradient flow trajectories reach $S$ rather than wander off.
\item[(J3)] It supplies a specific \emph{quantitative} global bound,
$\|\nabla f\| \ge \sqrt{2\mu h}$, tying the size of the gradient at height $h$
to $h$ itself for all $h$, however large.
\end{enumerate}

The thesis of this note is that (J1) and (J2) carry all the structural
conclusions of \cite{BCR2026}, while (J3) carries only the quantitative ones.
Accordingly we work with the following condition, in which $\alpha$ is a
comparison function in the sense familiar from stability theory.

\begin{defn}[Comparison function]\label{def:alpha}
A function $\alpha\colon[0,\infty)\to[0,\infty)$ is \emph{positive definite} if
it is continuous, $\alpha(0)=0$, and $\alpha(s)>0$ for all $s>0$.
\end{defn}

\begin{defn}[Gradient domination]\label{def:GD}
Let $\alpha$ be positive definite. We say $f$ is \emph{$\alpha$-gradient
dominated} if
\begin{equation}
  \|\nabla f(x)\| \ \ge\ \alpha\bigl(h(x)\bigr) \qquad \text{for all } x \in \M .
  \tag{$\mathrm{GD}_\alpha$}\label{eq:GD}
\end{equation}
\end{defn}

\begin{defn}[Condition (A)]\label{def:A}
A positive definite $\alpha$ satisfies \emph{condition {\rm(A)}} if there exist
$\mu>0$ and $h_0>0$ such that
\begin{equation}
  \alpha(s)^2 \ \ge\ 2\mu s \qquad \text{for all } s \in [0,h_0] .
  \tag{A}\label{eq:A}
\end{equation}
Equivalently, $\liminf_{s\downarrow 0} \alpha(s)^2/s > 0$.
\end{defn}

Global \eqref{eq:PL} is the special case $\alpha(s)=\sqrt{2\mu s}$. We call
$\alpha$ a \emph{witness} for $f$; it is never unique, and we shall be free to
shrink it.

A word on notation, since the reader of \cite{BCR2026} will already have
$\varphi$ and $\psi$ spoken for by the diffeomorphisms produced there. We write
$\alpha$ for the comparison function throughout; readers coming from control
theory will recognize the notation, and the letter is otherwise unused in
\cite{BCR2026}. We also freely write $\bcr$ for that paper when citing its
numbered results, so that ``$\bcr$ Theorem 1.2'' means Theorem 1.2 of
\cite{BCR2026}.

\subsection{Summary of what we claim}

Under \eqref{eq:GD} with $\alpha$ positive definite and satisfying \eqref{eq:A},
every structural result of \cite{BCR2026} holds \emph{verbatim}: the same
statement, with the same proof modulo three explicit substitutions which we carry
out in Section~\ref{sec:preliminaries}. In particular $f = f^* + \|\phi\|^2$ on
contractible $\M$, and all of the corollaries about $S$ follow. Table~\ref{tab:status}
summarizes which results survive.

There are in fact two routes to those conclusions. The short one shrinks the
witness so as to reparametrize $f$ into a function satisfying \eqref{eq:PL}
itself, and
then quotes \cite{BCR2026} verbatim (Proposition~\ref{prop:shortcut}); the direct
one replaces three steps of their proof, and additionally yields estimates in terms
of the original witness. We give both, and say in Section~\ref{sec:shortcut} what
each is good for.

What fails is exactly the quantitative layer built on top of global quadratic
growth: the inequality $h \ge \tfrac{\mu}{2}\dist(\cdot,S)^2$ now holds only near
$S$. Section~\ref{sec:fails}
makes this precise, and Section~\ref{sec:examples} shows that the hypothesis is
sharp: neither the behavior required at the origin nor the level-wise uniformity
can be dropped, and in each case the failure is visible in an elementary example
on $\R$ or $\R^2$.

Perhaps the most satisfying way to say what is going on is this. Set
\begin{equation}\label{eq:Psi}
  \Psi(h) \ :=\ \int_0^h \frac{ds}{\alpha(s)} .
\end{equation}
Then \eqref{eq:GD} is a global desingularized gradient inequality of
Kurdyka--\L{}ojasiewicz type \cite{Loj1982, Kurdyka1998}, with desingularizing
function $\Psi$, since $\Psi'(h)\|\nabla f\| \ge 1$ off $S$; and \eqref{eq:A}
gives $\Psi(h) = O(\sqrt h)$ at the origin. So the theorem being generalized
reads: \emph{a smooth function on a contractible complete manifold satisfying a
global inequality of this type, whose desingularizer is $O(\sqrt{h})$ at the
origin, is a nonlinear sum of squares.} The global P\L{} condition is the case
$\Psi(h)=\sqrt{2h/\mu}$ exactly. We write ``of \KL{} type'' rather than ``a \KL{}
inequality'' because the standard formulations often require the desingularizer to
be concave, which here would amount to requiring $\alpha$ to be nondecreasing, and
we do not.

\subsection{Relation to the comparison-function hierarchy}
\label{sec:hierarchy}

It is worth locating our hypothesis inside the hierarchy of
\cite{Sontag2025L4DC} before studying it, because it
turns out to coincide exactly with one of the classes already named there.

Recall the standard comparison classes for a continuous
$\alpha\colon\R_{\ge0}\to\R_{\ge0}$ with $\alpha(0)=0$: such an $\alpha$ is
\emph{positive definite} ($\mathcal{PD}$) if $\alpha(r)>0$ for $r>0$; of
\emph{class $\mathcal{K}$} if moreover strictly increasing; and of \emph{class
$\mathcal{K}_\infty$} if moreover unbounded.\footnote{We follow the definition in
\cite{CuiJiangSontag2024}, where a class $\mathcal{K}$ function is required to be
continuous, strictly increasing and vanishing at the origin, as is standard;
\cite{Sontag2025L4DC} states the requirement as ``nondecreasing''. Nothing below
depends on which convention is adopted, since all we ever use is the inclusion
$\mathcal{K}\subset\mathcal{PD}$.} Recall also the four gradient
dominance conditions on $f$ (we state them directly as conditions on $f$, in the
form given in \cite[\S2]{Sontag2025L4DC}):
\begin{align*}
  gl\text{-P\L{}I}  &: \ \exists\, c>0 \ \text{ s.t. } \ \|\nabla f\|^2 \ge c\,h \ \text{ on } \M ,\\
  sat\text{-P\L{}I} &: \ \exists\, a,b>0 \ \text{ s.t. } \ \|\nabla f\|^2 \ge \frac{a\,h}{b+h} \ \text{ on } \M ,\\
  sgl\text{-P\L{}I} &: \ \forall\,\rho>0 \ \ \exists\, c_\rho>0 \ \text{ s.t. } \ \|\nabla f\|^2 \ge c_\rho\, h \ \text{ on } \{h\le\rho\} ,\\
  loc\text{-P\L{}I} &: \ \exists\,\rho>0 \ \ \exists\, c_\rho>0 \ \text{ s.t. } \ \|\nabla f\|^2 \ge c_\rho\, h \ \text{ on } \{h\le\rho\} ,
\end{align*}
with the implications
\[
  gl\text{-P\L{}I} \ \Longrightarrow\ sat\text{-P\L{}I} \ \Longrightarrow\
  sgl\text{-P\L{}I} \ \Longrightarrow\ loc\text{-P\L{}I} ,
\]
and, on the level of the witness $\alpha$,
$gl \Rightarrow \mathcal{K}_\infty$, $sat \Rightarrow \mathcal{K}$,
$sgl \Rightarrow \mathcal{PD}$.

\begin{prop}[Our hypothesis is $sgl$-P\L{}I]\label{prop:sgl}
Let $f\colon\M\to\R$ be $C^1$ and bounded below. The following are equivalent:
\begin{enumerate}
\item[(i)] $f$ satisfies \eqref{eq:GD} for some positive definite $\alpha$
satisfying \eqref{eq:A};
\item[(ii)] $f$ satisfies $sgl$-P\L{}I.
\end{enumerate}
\end{prop}

\begin{proof}
(i)$\Rightarrow$(ii). Fix $\rho>0$; we may assume $\rho>h_0$. On $\{h\le h_0\}$,
\eqref{eq:A} gives $\|\nabla f\|^2\ge\alpha(h)^2\ge 2\mu h$. On
$\{h_0\le h\le\rho\}$, let $m:=\min_{[h_0,\rho]}\alpha>0$, which is positive
because $\alpha$ is continuous and strictly positive on the compact interval
$[h_0,\rho]$; then $\|\nabla f\|^2\ge m^2\ge (m^2/\rho)h$. So
$c_\rho:=\min\{2\mu, m^2/\rho\}$ works.

(ii)$\Rightarrow$(i). Here we must exhibit a single $\alpha$ that is positive
definite, satisfies \eqref{eq:GD}, and satisfies \eqref{eq:A}.

Dispose first of the constant case: if $h\equiv0$ then \eqref{eq:GD} holds for
$\alpha(s)=\sqrt s$, which is positive definite and satisfies \eqref{eq:A}. (This
case must be separated out, because for constant $f$ every positive constant is
admissible in (ii) and the supremum below is infinite.) So assume $f$ is not
constant.

Fix any $h_0>0$, and let $c(\rho)$ denote the supremum of the constants
admissible in (ii) for a given $\rho$; a constant admissible for $\rho$ is
admissible for every smaller value, so $c$ is nonincreasing, and it is positive
and finite by hypothesis. Put
\[
  \gamma(s) \ :=\ \int_1^2 c(sv)\,dv \ =\ \frac1s\int_s^{2s} c(u)\,du , \qquad
  \tilde\gamma(s) \ :=\ \min\{\gamma(s), \gamma(h_0)\} .
\]
Then $\gamma$ is nonincreasing (because $c$ is), continuous (the second
expression exhibits it as an average of a monotone, hence locally integrable,
function over an interval varying continuously with $s$), and satisfies
$\gamma\le c$ pointwise. Hence $\tilde\gamma$ is continuous, positive, bounded by
$\gamma(h_0)$, and $\le c$. Setting
$\alpha(s):=\sqrt{\tilde\gamma(s)\,s}$ gives a continuous function with
$\alpha(s)\le\sqrt{\gamma(h_0)s}\to0$ as $s\downarrow0$, so $\alpha(0)=0$; it is
strictly positive for $s>0$, hence positive definite; it satisfies \eqref{eq:GD}
since $\|\nabla f\|^2\ge c(h)h\ge\tilde\gamma(h)h=\alpha(h)^2$; and for
$s\le h_0$ monotonicity of $\gamma$ gives $\tilde\gamma(s)=\gamma(h_0)$, so
\eqref{eq:A} holds with $2\mu=\gamma(h_0)$.
\end{proof}

So the main theorem of this note can be stated without any reference to a witness
$\alpha$ at all: \emph{the structure theory of \cite{BCR2026} holds under
$sgl$-P\L{}I}. Because $gl \Rightarrow sat \Rightarrow sgl$, it holds a fortiori
under the saturated condition, which is the one available for continuous-time
LQR.

Two warnings about how this fits into the diagram, both of which we think are
worth making explicit.

\begin{rem}[The two hierarchies are transverse]\label{rem:transverse}
The chain $gl\Rightarrow sat\Rightarrow sgl\Rightarrow loc$ and the chain
$\mathcal{K}_\infty\Rightarrow\mathcal{K}\Rightarrow\mathcal{PD}$ are not two
gradings of the same axis. The P\L{}I conditions constrain $\alpha$ near the
origin (and, for $gl$ and $sat$, uniformly); the class conditions constrain
$\alpha$ at infinity. Our hypothesis lives entirely on the first axis, and it is
incomparable with the second:
\begin{itemize}
\item $\mathcal{K}_\infty \not\Rightarrow$ our hypothesis. Take
$\alpha(r)=r$, which is of class $\mathcal{K}_\infty$ but violates \eqref{eq:A}.
It is attained: $f(x)=e^{-x}$ on $\R$ satisfies \eqref{eq:GD} with equality for
this $\alpha$ and has \emph{no minimizer at all} (Example~\ref{ex:exp}). So even
the strongest class condition, the one that yields full ISS in
\cite{Sontag2025L4DC}, implies nothing about the structure of $f$.
\item Our hypothesis $\not\Rightarrow\mathcal{K}$. The function
$f(x)=\log(1+\|x\|^2)$ on $\R^n$ is $sgl$-P\L{}I but admits no nondecreasing
positive witness, since its sharpest $\alpha$ tends to $0$ at infinity
(Example~\ref{ex:log}).
\item The gap is populated by conditions of independent interest. The class
studied by Fatkhullin and Polyak \cite{FatkhullinPolyak2021},
\[
  \alpha(r) \ =\ \sqrt{\frac{ar}{(b+r)^2}} ,
\]
is $\mathcal{PD}$ but not $\mathcal{K}$, since it decays at infinity; yet
$\alpha(r)^2 = ar/(b+r)^2 \approx (a/b^2)r$ near the origin, so it satisfies
\eqref{eq:A} and is $sgl$-P\L{}I. A function in that class therefore enjoys the
\emph{entire} structure theory below while sitting on the lowest rung of the
robustness ladder, where only integral ISS is available.
\end{itemize}
In other words, ISS-type robustness and the nonlinear-least-squares normal form
are logically independent properties, governed by opposite ends of the same
inequality. The same goes for the stochastic ladder of \cite{CuiJiangSontag2026},
which is indexed by the same three classes; see Section~\ref{sec:applications}.
This is the same phenomenon studied from the other side in
Section~\ref{sec:fails}: the structure theory is insensitive to the global size
of the gradient, which is exactly what the class conditions measure.
\end{rem}

\begin{rem}[Two readings of ``semiglobal'']\label{rem:twosgl}
The condition $sgl$-P\L{}I is stated in \cite[\S2]{Sontag2025L4DC} in two forms:
a \emph{sublevel} form, quantified over $\rho>0$ as above, and a \emph{compact}
form, quantified over compact subsets $C\subset D$ of the domain. The two agree
when the sublevel sets $\{h\le\rho\}$ are compact, which is the case in the LQR
application, where the loss blows up at $\partial D$, but in general the
sublevel form is strictly stronger, and it is the sublevel form that we use.
Example~\ref{ex:tanh} shows the difference is not academic: there is a smooth
function on $\R$, with a single nondegenerate minimizer, satisfying the compact
form and $loc$-P\L{}I, for which the conclusions of \cite{BCR2026} fail.
\end{rem}

Finally, all of the inclusions below are strict, with the separating examples
collected in Section~\ref{sec:examples}:
\[
  \left\{\begin{array}{c}\text{strongly}\\\text{convex}\end{array}\right\}
  \subsetneq
  \Bigl\{gl\text{-P\L{}I}\Bigr\}
  \subsetneq
  \Bigl\{sat\text{-P\L{}I}\Bigr\}
  \subsetneq
  \Bigl\{sgl\text{-P\L{}I}\Bigr\}
  \subsetneq
  \Bigl\{loc\text{-P\L{}I}\Bigr\} ,
\]
and our theorems hold under $sgl$-P\L{}I, hence under $sat$-P\L{}I and
$gl$-P\L{}I as well, but not in general under $loc$-P\L{}I.
The condition
\[
  \|\nabla f(x)\|^2 \ \ge\ \frac{2\mu\, h(x)}{1+h(x)}
\]
which motivated this note is $sat$-P\L{}I with $a=2\mu$, $b=1$; see
Example~\ref{ex:bounded}.

\section{The machinery of \cite{BCR2026} that never mentions P\L{}}
\label{sec:free}

Before doing any work, it pays to notice how much of \cite{BCR2026} is already
free of the P\L{} hypothesis. The entire single-minimizer analysis is stated in
terms of coercivity:

\begin{thm}[$\bcr$ Theorem 3.3]\label{thm:BCR33}
Let $f\colon\M\to\R$ be smooth and coercive (compact sublevel sets). Assume $f$
has a unique critical point $x^*$ and that the Hessian of $f$ at $x^*$ is
positive definite. Then there exists a diffeomorphism $\phi\colon\M\to\R^n$
such that $f(x) = f(x^*) + \|\phi(x)\|^2$ for all $x\in\M$.
\end{thm}

Theorem~\ref{thm:BCR33} and everything feeding it: $\bcr$ Lemmas 3.2, 3.4,
3.5, 3.6, and Appendices A--D, that is, the globalized Morse lemma, the
Palais--Cerf extension, the rescaled gradient flow, and the transporter,
never invoke \eqref{eq:PL}. Likewise the topological input ($\bcr$
Theorem 1.3, Appendices E, F, G) and the construction of P\L{} functions
($\bcr$ Theorem 1.5) are untouched by anything we do here.

Consequently, the direct route to generalizing \cite{BCR2026} amounts to
re-deriving, from \eqref{eq:GD} and \eqref{eq:A}, exactly three key facts (the
shorter route of Section~\ref{sec:shortcut} needs none of them, but does not
preserve the quantitative estimates in terms of the original witness):
\begin{enumerate}
\item[(I1)] a bound on the \emph{length} of negative gradient trajectories in
terms of the initial function value alone ($\bcr$ Lemma 2.1);
\item[(I2)] the Morse--Bott property at $S$ ($\bcr$ Lemma 2.2);
\item[(I3)] coercivity and nondegeneracy of $f$ restricted to a fiber of the
end-point map ($\bcr$ Proposition 4.4).
\end{enumerate}
That is the whole content of Section~\ref{sec:preliminaries}. It is a little
striking how narrow the interface is.

\section{Preliminaries under $(\mathrm{GD}_\alpha)$ and {\rm(A)}}
\label{sec:preliminaries}

Throughout this section, $\alpha$ is positive definite and satisfies
\eqref{eq:A} with constants $\mu, h_0$, and $\Psi$ is as in \eqref{eq:Psi}.

\subsection{The desingularizer}

\begin{lem}[Properties of $\Psi$]\label{lem:Psi}
Let $\alpha$ be positive definite and satisfy \eqref{eq:A}. Then:
\begin{enumerate}
\item[(i)] $\Psi(h) < \infty$ for every $h \ge 0$.
\item[(ii)] $\Psi$ is a continuous, strictly increasing bijection from $[0,\infty)$
onto $[0,L)$, where $L := \lim_{h\to\infty}\Psi(h) \in (0,\infty]$. Its
inverse $\Psi^{-1}$, mapping $[0,L)$ to $[0,\infty)$, is continuous and
strictly increasing, and $\Psi(0)=\Psi^{-1}(0)=0$.
\item[(iii)] $\Psi(h) \le \sqrt{2h/\mu}$ for $0 \le h \le h_0$; equivalently
$\Psi^{-1}(d) \ge \tfrac{\mu}{2}d^2$ for $0 \le d \le \Psi(h_0)$.
\item[(iv)] Any $f$ satisfying \eqref{eq:GD} obeys the \emph{local} P\L{} inequality
\[
  h(x) \ \le\ \frac{1}{2\mu}\,\|\nabla f(x)\|^2
  \qquad\text{on}\qquad \Omega_0 := \{x\in\M : h(x) \le h_0\} ,
\]
and $\Omega_0$ contains an open neighborhood of $S = \{h = 0\}$.
\end{enumerate}
\end{lem}

\begin{proof}
The only possible obstruction to (i) is the endpoint $s\to 0^+$, because $\alpha$
is continuous and strictly positive on $(0,\infty)$, hence bounded below by a
positive constant on each compact subinterval $[\varepsilon, h]$; so
$\int_\varepsilon^h ds/\alpha(s) < \infty$. Near $0$, \eqref{eq:A} gives
$1/\alpha(s) \le 1/\sqrt{2\mu s}$, which is integrable, and integrating this
bound on $[0,h]$ with $h\le h_0$ yields (iii). Item (ii) is immediate from
(i) and $1/\alpha > 0$ on $(0,\infty)$. Item (iv) is \eqref{eq:GD} combined
with \eqref{eq:A}.
\end{proof}

Observe that property (iv) is exactly the statement that our condition
implies a local P\L{} inequality, with the \emph{same} constant $\mu$, on a
sublevel-set neighborhood of $S$. Everything in (J1) will be read off from it.

\begin{rem}[The sharpest witness]\label{rem:sharpest}
Given $f$, put
\[
  \alpha_f(s) \ :=\ \inf\bigl\{ \|\nabla f(x)\| \ :\ x\in\M,\ h(x) = s \bigr\}
  \in [0,\infty] ,
\]
with the convention $\inf\emptyset = +\infty$. Then $f$ satisfies \eqref{eq:GD}
for \emph{some} positive definite $\alpha$ if and only if
\[
  \inf\bigl\{\|\nabla f(x)\| : a \le h(x) \le b\bigr\} > 0
  \qquad \text{for all } 0 < a \le b < \infty ,
\]
in which case a continuous positive definite minorant of $\alpha_f$ can be
constructed by a routine interpolation across the dyadic scales; and
\eqref{eq:A} then reads $\liminf_{s\downarrow 0}\alpha_f(s)^2/s > 0$. We mention
this only to make clear that the hypothesis is a property of $f$ and not of a
chosen $\alpha$; in practice one exhibits a convenient $\alpha$ and moves on.
\end{rem}

\begin{rem}[Uniform invexity]\label{rem:invex}
Recall that $f$ is \emph{invex} if $\nabla f(x)=0 \Rightarrow f(x)=f^*$. In the
language of Remark~\ref{rem:sharpest}, invexity says that no level set with
$s>0$ \emph{contains} a critical point, whereas positive definiteness of
$\alpha_f$ says the gradient is bounded away from zero on each such level set.
The gap between the two is real when level sets are noncompact: the infimum may
vanish without being attained, with the gradient dying only at infinity. Our
proofs use the uniform version, through the finiteness of $\Psi$, and this is not
an artifact: Example~\ref{ex:tanh} exhibits a smooth invex $f$ on $\R$ with a
single nondegenerate minimizer, satisfying $loc$-P\L{}I and with $\alpha_f>0$
throughout the range of $h$, for which the conclusion of Theorem~\ref{thm:one}
fails. There, $\alpha_f$ decays to zero as $h$ approaches $\sup f<\infty$, so no
positive definite witness exists and $\Psi$ is infinite.

That example leaves one narrower question open. Suppose $h(\M)=[0,\infty)$,
which asks both that $f$ attain its infimum and that it be unbounded above, and
suppose $\alpha_f(r)>0$ for every $r>0$ but $\inf_{[a,b]}\alpha_f = 0$ for some
band $0<a<b<\infty$ (possible, since $\alpha_f$ need not be lower
semicontinuous). Then no continuous positive definite witness exists, so the
desingularizer $\Psi$ of \eqref{eq:Psi} is unavailable and our argument breaks
down; but we know of no counterexample to the conclusions.

Let us also note that plain invexity is definitely not enough, as $\bcr$
footnote 4 already shows: $f(x,y)=(x^2y-x-1)^2+(x^2-1)^2$ is smooth and invex
with disconnected $S$.
\end{rem}

\subsection{(I1): trajectories, and where the minimizers are}

We open with a small lemma that will be used three times below. In each of those
places two negative gradient flows
differ by a positive factor depending only on the function value, and we want to
conclude that they share not merely their orbits but their end-point maps. The
first half is immediate and the second is not, since a positive factor can in
principle traverse only part of an orbit in infinite time. What rules that out is
that $h$ decreases along the flow, so the factor is evaluated on a compact
interval. We isolate the argument once. Note that it uses neither \eqref{eq:GD}
nor \eqref{eq:A}, so invoking it later costs no logical dependence on the rest of
this section. The Lipschitz hypothesis on $\sigma$ is only what makes the flow of
the rescaled field unique, and it is met in all three applications, where $\sigma$
is $\theta'$, $e^{-2\chi}$ and $\theta'$ again, each of them smooth.

\begin{lem}[Positive reparametrizations of the gradient flow]\label{lem:reparam}
Let $f\colon\M\to\R$ be smooth and bounded below, $h=f-f^*$, and let
$\sigma\colon[0,\infty)\to(0,\infty)$ be locally Lipschitz. Consider the two vector
fields
\[
  X \ =\ -\nabla f , \qquad Y \ =\ -\sigma(h)\,\nabla f .
\]
Then $X$ and $Y$ have the same zeros. Moreover, fix $x_0\in\M$ and let
$x\colon[0,T)\to\M$ and $y\colon[0,S)\to\M$ be the maximal forward trajectories of
$X$ and of $Y$ issuing from $x_0$. There is an increasing bijection
$t\colon[0,S)\to[0,T)$ with $y(s)=x(t(s))$, and it satisfies
\[
  m \ \le\ t'(s) \ \le\ M \qquad\text{for all } s ,
\]
where $m$ and $M$ are the minimum and maximum of $\sigma$ on $[0,h(x_0)]$. In
particular $T=\infty$ if and only if $S=\infty$, and in that case
$\lim_{t\to\infty}x(t)$ exists if and only if $\lim_{s\to\infty}y(s)$ does, with
the same value. Hence, whenever the end-point maps are defined, they agree.
\end{lem}

\begin{proof}
The zeros coincide because $\sigma>0$. Both $h\circ x$ and $h\circ y$ are
nonincreasing, since $\tfrac{d}{ds}h(y(s)) = -\sigma(h)\|\nabla f\|^2 \le 0$ and
likewise for $x$; so both trajectories remain in $\{0\le h\le h(x_0)\}$. As
$\sigma$ is continuous and strictly positive on the compact interval
$[0,h(x_0)]$, it satisfies $0<m\le\sigma\le M<\infty$ there.

Define $t(s) = \int_0^s \sigma(h(y(u)))\,du$, so that $t(0)=0$ and
$t'(s)=\sigma(h(y(s)))\in[m,M]$. Then $s\mapsto y(s)$ and $s\mapsto x(t(s))$
satisfy the same ordinary differential equation
$\tfrac{d}{ds} = -\sigma(h)\nabla f$ with the same initial condition. Its right-hand
side is locally Lipschitz, being the product of the locally Lipschitz
$\sigma\circ h$ with the smooth $\nabla f$, so solutions are unique and the two
agree wherever both are defined. Since $t$ is a bi-Lipschitz increasing
bijection onto its image, maximality of the two trajectories forces
$t([0,S)) = [0,T)$. From $ms\le t(s)\le Ms$ we get $T=\infty$ if and only if
$S=\infty$. The two trajectories then have the same image traversed in the same
order, so one converges if and only if the other does, to the same point.
\end{proof}

The following replaces $\bcr$ Lemma 2.1. Note that the conclusion is not
quadratic growth but a growth estimate governed by $\Psi$; quadratic growth
survives only locally.

\begin{lem}[Bounded trajectories and growth]\label{lem:traj}
Let $f\colon\M\to\R$ be smooth, bounded below, and satisfy \eqref{eq:GD} with
$\alpha$ positive definite satisfying \eqref{eq:A}. Let $x(\cdot)$ solve
$x'(t) = -\nabla f(x(t))$, $x(0)=x_0$. Then $x(t)$ is defined for all $t\ge 0$,
the trajectory has length at most $\Psi(h(x_0))$ on $[0,\infty]$, and it
converges to a limit $x_\infty := \lim_{t\to\infty} x(t)$ satisfying
\[
  \nabla f(x_\infty) = 0, \qquad f(x_\infty) = f^*, \qquad
  \dist(x_0, x_\infty) \ \le\ \Psi\bigl(h(x_0)\bigr) .
\]
Consequently $f$ attains its infimum, the set
\[
  S \ :=\ \{x : \nabla f(x)=0\} \ =\ \{x : f(x) = f^*\}
\]
is non-empty and closed, and for all $x\in\M$
\begin{equation}\label{eq:growth}
  \dist(x,S) \ \le\ \Psi\bigl(h(x)\bigr) < L ,
  \qquad\text{equivalently}\qquad
  h(x) \ \ge\ \Psi^{-1}\bigl(\dist(x,S)\bigr) .
\end{equation}
In particular, quadratic growth holds locally:
\begin{equation}\label{eq:QGloc}
  h(x) \ \ge\ \frac{\mu}{2}\,\dist(x,S)^2
  \qquad \text{whenever } h(x) \le h_0 .
  \tag{QG$_{\text{loc}}$}
\end{equation}
\end{lem}

\begin{proof}
Let $[0,T)$ be the maximal interval of existence and put
$\eta(t) := h(x(t))$, so $\eta' = -\|\nabla f(x(t))\|^2 \le 0$. If
$\nabla f(x(t_1))=0$ for some $t_1$, then $x$ is constant on $[t_1,T)$ by
uniqueness, and it suffices to argue on $[0,t_1)$; so assume that
$\nabla f(x(t))\ne 0$, hence that $\eta$ is strictly decreasing, on the
interval considered. For $t$ in that interval, changing variables to $s = \eta(\tau)$,
\[
  \int_0^{t}\|\nabla f(x(\tau))\|\,d\tau
  = \int_0^{t}\frac{-\eta'(\tau)}{\|\nabla f(x(\tau))\|}\,d\tau
  \le \int_0^{t}\frac{-\eta'(\tau)}{\alpha(\eta(\tau))}\,d\tau
  = \int_{\eta(t)}^{\eta(0)}\frac{ds}{\alpha(s)}
  \ \le\ \Psi\bigl(h(x_0)\bigr) ,
\]
finite by Lemma~\ref{lem:Psi}(i). Thus the trajectory remains in the closed ball
of radius $\Psi(h(x_0))$ about $x_0$, which is compact since $\M$ is complete
(Hopf--Rinow), so $T=\infty$ by the escape lemma. Finite length makes $t\mapsto x(t)$
Cauchy, and completeness gives convergence to some $x_\infty$ at distance at most
$\Psi(h(x_0))$ from $x_0$.

Since $\int_0^\infty\|\nabla f(x(t))\|^2\,dt = h(x_0) - \lim_t \eta(t) < \infty$,
we have $\liminf_{t\to\infty}\|\nabla f(x(t))\|=0$, whence
$\nabla f(x_\infty)=0$ by continuity along a subsequence. Then \eqref{eq:GD}
gives $\alpha(h(x_\infty)) \le 0$, so $h(x_\infty)=0$ by positive definiteness
of $\alpha$: the limit is a global minimizer, and $f$ attains $f^*$. The same
argument applied at an arbitrary critical point shows
$\{\nabla f = 0\} = \{f = f^*\}$, a closed set. Estimate \eqref{eq:growth} follows
by taking $x_0 = x$, and \eqref{eq:QGloc} follows from
Lemma~\ref{lem:Psi}(iii).
\end{proof}

Observe that the argument proves that $f$ attains its infimum, rather than
assuming it; positive definiteness of $\alpha$ away from the origin is what
prevents the trajectory from sliding off to infinity, and \eqref{eq:A} is what
makes $\Psi$ finite. Both halves are needed, as Examples~\ref{ex:exp}
and~\ref{ex:x4} show.

\begin{rem}[$L$ is infinite unless $f$ is constant]\label{rem:L}
By \eqref{eq:growth}, every point of $\M$ lies within distance $L$ of $S$, so if
$L<\infty$ then $\M$ is a bounded tube around $S$. In fact this cannot happen at
all except in the trivial case: \emph{if $f$ is not constant then $L=\infty$.}
Indeed, suppose $L<\infty$, fix $x\in S$ and let $F=\pi^{-1}(x)$. Every $y\in F$
has $\pi(y)=x$, so \eqref{eq:growth} gives $\dist(y,x)\le\Psi(h(y))<L$ and $F$ is
bounded; $F$ is closed in $\M$, being properly embedded
(Proposition~\ref{prop:fiber} below), so $F$ is compact because $\M$ is complete. But
$F$ is diffeomorphic to $\R^k$, and $k\ge1$ whenever $f$ is not constant (if
$k=0$ then $S$ is open as well as closed in the connected manifold $\M$, so
$S=\M$), a contradiction. For global P\L{}, $L=\infty$ directly.
\end{rem}

\begin{rem}[Coercivity in place of completeness]\label{rem:coercive}
Completeness of $\M$ is used in the proof above only to know that closed balls
are compact, and hence to rule out escape in finite time. If instead one assumes
that $f$ is \emph{coercive} on $\M$ (compact sublevel sets), the same conclusions
follow, and more easily: the trajectory from $x_0$ stays in the compact set
$\{f\le f(x_0)\}$, so it cannot escape, and finite length gives convergence
inside that set. The same substitution works in the two other places where
completeness appears downstream, namely the coercivity of $f|_F$ in
Proposition~\ref{prop:fiber} (which we prove directly in any case) and the escape
argument in Theorem~\ref{thm:bundle}, where the set $B\cap C$ may be replaced by
$\{f=\bar f\}\cap\pi^{-1}(c([0,1]))$, compact because $f$ is coercive.

This matters for the applications in Section~\ref{sec:applications}, where the
natural domain is an open subset $D\subsetneq\R^{n}$ (the stabilizing feedback
gains, say) which is not complete in the ambient Euclidean metric but on which
the loss is coercive because it blows up at $\partial D$. In the language of
\cite{Sontag2025L4DC, Sontag2022}, this is the requirement that $f-f^*$ be a
\emph{size function}, or barrier metric, for the pair $(D,S)$: continuous,
positive definite with respect to $S$, and proper. Note that the same hypothesis
appears on both sides: our growth estimate \eqref{eq:growth} is exactly the
assertion that $h$ is a proper exhaustion relative to $S$, expressed in a way that
does not refer to $\partial D$.

The results of \cite{BCR2026} that we import wholesale (above all
Theorem~\ref{thm:BCR33}) are stated under that paper's standing convention that
$\M$ is complete. However, this can be weakened, because when $f$ is coercive, a
conformal change of metric reduces the coercive case to the complete case.

Let us spell this out. Write $g$ for the given Riemannian metric on $\M$, that is,
for the smoothly varying inner product $g_x(\cdot,\cdot)$ on each tangent space,
and given a smooth function $u\colon\M\to\R$ let $e^{2u}g$ denote the metric whose
inner product at $x$ is $e^{2u(x)}g_x(\cdot,\cdot)$: a \emph{conformal rescaling}
of $g$, which stretches all lengths at $x$ by the same factor $e^{u(x)}$
regardless of direction. Two elementary consequences will be used, both immediate
from the definition of the gradient by $g(\nabla_g f,\cdot) = df$:
\[
  \nabla_{e^{2u}g}\,f = e^{-2u}\,\nabla_g f ,
  \qquad
  \|\nabla_{e^{2u}g}\,f\|_{e^{2u}g} = e^{-u}\,\|\nabla_g f\|_g .
\]
Now put $g' = e^{2\chi(h)}g$, where $h=f-f^*$ as usual and $\chi$ is a smooth
nondecreasing function to be chosen, and observe three things.

\emph{(a) The hypothesis survives.} On each band $\{a\le h\le b\}$, which is
compact by coercivity, the factor $e^{-\chi(h)}$ is bounded below by a positive
constant, so the rescaled witness $\alpha'(s) := e^{-\chi(s)}\alpha(s)$ is again
positive definite and
$\|\nabla_{g'} f\|_{g'} = e^{-\chi(h)}\|\nabla_g f\|_g \ge \alpha'(h)$ obeys a
positive definite bound on every level; and \eqref{eq:A} survives too, being a
statement on $\{h\le h_0\}$, where $e^{-\chi(h)}\ge e^{-\chi(h_0)}>0$. So
$sgl$-P\L{}I holds for $g'$.

\emph{(b) The conclusions are unaffected.} Since
$\nabla_{g'} f = e^{-2\chi(h)}\nabla_g f$ is a positive smooth multiple of
$\nabla_g f$, with the factor $e^{-2\chi(h)}$ depending on $x$ only through $h$,
Lemma~\ref{lem:reparam} applies with $\sigma = e^{-2\chi}$: the two negative
gradient flows have the same critical set $S$, the same orbits, and the same
end-point map $\pi$, hence the same fibers. The \emph{horizontal-transport
map} used to trivialize $\pi$ is unchanged as well: it is built from horizontal
lifts, and ``horizontal'' means orthogonal to $\ker D\pi$, a condition insensitive
to a conformal factor, since $g'(u,v)=0$ if and only if $g(u,v)=0$. The horizontal
lift of a curve is determined by that distribution alone (given it, $\gamma'(t)$
is the unique horizontal vector projecting to $c'(t)$) so that map is the same
for $g$ and $g'$. The remaining ingredient, the diffeomorphism carrying one fiber
to $\R^k$ furnished by Theorem~\ref{thm:BCR33}, need not literally be the same for
the two metrics; only its existence is ever used, and the final conclusion
$f = f^*+\|\phi\|^2$ refers to no metric at all.

\emph{(c) $\chi$ can be chosen so that $g'$ is complete.} Put
$M(s):=\max\{\|\nabla_g f\|_g : h\le s\}$, which is finite because $\{h\le s\}$ is
compact, and nondecreasing in $s$ by construction; choose any smooth
nondecreasing $\chi$ with $e^{\chi(s)}\ge M(s)(1+s)$, which exists because a
nondecreasing finite function admits a smooth nondecreasing majorant. Let $\gamma$
be a curve leaving every compact set; then $h\to\infty$ along it, by coercivity.
Since $|(h\circ\gamma)'|\le\|\nabla_g f\|_g\|\gamma'\|_g$, since
$\|\nabla_g f\|_g\le M(h)$ pointwise, and since $e^{\chi(s)}/M(s)\ge 1+s$, the
$g'$-length of $\gamma$ satisfies
\[
  L_{g'}(\gamma) \ =\ \int e^{\chi(h)}\|\gamma'\|_g\,dt
  \ \ge\ \int \bigl(1+h\bigr)\bigl|(h\circ\gamma)'\bigr|\,dt
  \ \ge\ \bigl|W(h(\gamma(t))) - W(h(\gamma(t_0)))\bigr| ,
\]
where $W(s) := s+\tfrac{s^2}{2}$, the last step because $W'=1+s$. Note that no monotonicity of $h\circ\gamma$ is
needed. Since $h(\gamma(t))\to\infty$, the right side tends to infinity, so no
curve of finite $g'$-length escapes and $g'$ is complete by Hopf--Rinow.
\end{rem}

\begin{lem}[Morse--Bott]\label{lem:MB}
Let $f\colon\M\to\R$ be smooth and satisfy \eqref{eq:GD} with $\alpha$ positive
definite satisfying \eqref{eq:A}, and let $S$ be its set of critical points.
Then each connected component of $S$ is a $C^\infty$ embedded submanifold of
$\M$, and for each $x\in S$,
\begin{equation}\label{eq:MB}
  \ker\nabla^2 f(x) = T_xS,
  \qquad
  \nabla^2 f(x)\big|_{N_xS} \succeq \mu\,\Id ,
\end{equation}
where $T_xS$ and $N_xS$ are the tangent and normal spaces at $x$ to the
component of $S$ through $x$.
\end{lem}

\begin{proof}
By Lemma~\ref{lem:Psi}(iv), $f$ satisfies the P\L{} inequality with parameter
$\mu$ on the neighborhood $\Omega_0 \supseteq S$ of $S$. The results of
Rebjock and Boumal \cite[Thm.~2.16, Cor.~2.17]{RB2024a}, quoted as $\bcr$
Lemma 2.2, require only a local P\L{} inequality, so they apply verbatim.
\end{proof}

This is the only place where \eqref{eq:A} is used, and it is used essentially:
see Example~\ref{ex:cross}.

For the applications of Section~\ref{sec:applications} we shall need the
converse direction, and it is worth isolating, because there is a small gap
between the two natural readings of ``local''. A nondegenerate minimizer gives a
P\L{} inequality on a \emph{neighborhood}, whereas $loc$-P\L{}I as we have defined
it asks for one on a \emph{sublevel set}; a neighborhood need not contain a
sublevel set if there are remote near-minimizers. Coercivity and uniqueness close
the gap.

\begin{lem}[Nondegenerate minimizer gives $loc$-P\L{}I]\label{lem:nondeg}
Let $f\colon\M\to\R$ be smooth and coercive, with a unique critical point $x^*$,
and suppose $\nabla^2f(x^*)\succ0$. Then $f$ satisfies $loc$-P\L{}I: there are
$\rho,c>0$ with $\|\nabla f\|^2 \ge c\,h$ on $\{h\le\rho\}$.
\end{lem}

\begin{proof}
Write $\lambda,\Lambda>0$ for the least and greatest eigenvalues of
$\nabla^2f(x^*)$. In normal coordinates at $x^*$, with $v=\Log_{x^*}(x)$, Taylor's
theorem gives $\nabla f(x) = \nabla^2f(x^*)[v]+o(\|v\|)$ and
$h(x) = \tfrac12\langle v,\nabla^2f(x^*)v\rangle+o(\|v\|^2)$, so for $\|v\|$ small
$\|\nabla f(x)\|\ge\tfrac{\lambda}{2}\|v\|$ and $h(x)\le\Lambda\|v\|^2$; hence
$\|\nabla f\|^2 \ge c\,h$ on some open neighborhood $U$ of $x^*$, with
$c=\lambda^2/(4\Lambda)$.

It remains to find $\rho>0$ with $\{h\le\rho\}\subseteq U$. If there were none,
there would be $x_j\notin U$ with $h(x_j)\to0$. Eventually $h(x_j)\le1$, and
$\{h\le1\}$ is compact by coercivity, so a subsequence converges to some
$\bar x\in\M\setminus U$, the complement of $U$ being closed. Then $h(\bar x)=0$,
so $\bar x$ is a global minimizer and hence a critical point, so $\bar x=x^*$. But
$x^*\in U$, a contradiction.
\end{proof}

\subsection{(I3): the end-point map and its fibers}

Everything now proceeds as in \cite{BCR2026}. Let $\Phi$ denote the negative
gradient flow of $f$, which by Lemma~\ref{lem:traj} is defined on
$\M\times[0,\infty)$, and let
\begin{equation}\label{eq:pi}
  \pi\colon\M\to S, \qquad \pi(y) := \lim_{t\to\infty}\Phi^t(y)
\end{equation}
be the end-point map.

\begin{prop}[$\bcr$ Propositions 4.2 and 4.3]\label{prop:pi}
Let $f\colon\M\to\R$ be smooth and satisfy \eqref{eq:GD} with $\alpha$ positive
definite satisfying \eqref{eq:A}. Then $\pi$ is continuous, $\pi(x)=x$ exactly
for $x\in S$, and $\M$ strongly deformation retracts to $S$; in particular $S$
is connected and $\M$ is contractible if and only if $S$ is. Moreover, $S$ is a
smooth, properly embedded submanifold of $\M$, and $\pi\colon\M\to S$ is a
smooth submersion.
\end{prop}

\begin{proof}
The proofs of $\bcr$ Propositions 4.2 and 4.3 apply once the single quantitative
ingredient they draw from $\bcr$ Lemma 2.1 is replaced. That ingredient appears
in the continuity argument, in the form: \emph{for each $x\in S$ and each
neighborhood $U'$ of $x$, there is a neighborhood $V$ of $x$ such that
$\Phi^s(z)\in U'$ for all $z\in V$ and all $s\ge0$.} To obtain it, choose
$\rho>0$ with $\rho<L$ and $B(x,2\rho)\subseteq U'$, and set
$V := B(x,\rho)\cap\{h < \Psi^{-1}(\rho)\}$, an open neighborhood of $x$ since
$h(x)=0$ and $\Psi^{-1}(\rho)>0$. For $z\in V$, Lemma~\ref{lem:traj} gives
$\dist(z,\Phi^s(z)) \le \Psi(h(z)) < \rho$, so $\Phi^s(z)\in B(x,2\rho)\subseteq U'$.
The remainder of $\bcr$ Proposition 4.2 (the deformation retraction, connectedness,
transfer of contractibility) is purely topological.

For $\bcr$ Proposition 4.3, smoothness of $S$ and property \eqref{eq:MB} come
from Lemma~\ref{lem:MB} here; the appeal to Falconer's theorem
\cite[Thm.~5.1]{Falconer1983} needs only that $g=\Phi^1$ is smooth on $\M$
(Lemma~\ref{lem:traj} and the fundamental theorem of flows) and that $S$ is
pseudo-hyperbolic for $g$, which follows from $Dg(x)=e^{-\nabla^2 f(x)}$ together
with \eqref{eq:MB}: $Dg(x)$ is the identity on $T_xS$ and has all eigenvalues in
$(0,e^{-\mu}]$ on $N_xS$. The submersion argument is unchanged, being local at
$S$ plus the identity $\pi = \pi\circ\Phi^t$.
\end{proof}

\begin{prop}[$\bcr$ Proposition 4.4]\label{prop:fiber}
Let $f$ be as in Proposition~\ref{prop:pi} and let $x\in S$. Then the fiber
$F := \pi^{-1}(x)$ is a contractible, properly embedded smooth submanifold of
$\M$, and, with the Riemannian submanifold structure, $f|_F$ is smooth, satisfies
\eqref{eq:GD} on $F$ with the same $\alpha$ and the same $f^*$, is coercive, and
has $x$ as its unique critical point, at which its Hessian is positive definite.
Consequently there is a diffeomorphism $\phi\colon F\to\R^k$,
$k = \dim\M - \dim S$, with $f|_F = f^* + \|\phi\|^2$.
\end{prop}

\begin{proof}
That $F$ is a contractible, properly embedded smooth submanifold, complete in the
induced metric, and that $\nabla(f|_F) = \nabla f|_F$ with critical set
$F\cap S = \{x\}$, is verbatim $\bcr$ Proposition 4.4. Since $x\in F$ and
$f(x)=f^*$ (Lemma~\ref{lem:traj}), we have $\inf_F f = f^*$, so the inequality
\eqref{eq:GD} for $f|_F$ on $F$ is literally the inequality for $f$ restricted
to $F$: the class is closed under passage to fibers. Coercivity is where the new
estimate enters: for $y\in F$ we have $\pi(y)=x$, so by Lemma~\ref{lem:traj}
$\dist_\M(y,x) \le \Psi(h(y))$; hence
$\{y\in F : f(y)\le f^*+c\}$ is bounded in $\M$ and closed in $\M$ ($F$ being
properly embedded), therefore compact, therefore compact in $F$. Applying
Lemma~\ref{lem:MB} to $f|_F$ on $F$, whose critical set is the single point $x$,
gives $\ker\nabla^2(f|_F)(x) = T_x\{x\} = \{0\}$, i.e.\ positive definiteness.
Now Theorem~\ref{thm:BCR33} applies to $f|_F$.
\end{proof}

\section{The theorems that survive}
\label{sec:survive}

We now discuss how the previous considerations lead to the main results.
Throughout this section, $f\colon\M\to\R$ is smooth
and bounded below and satisfies \eqref{eq:GD} for some positive definite
$\alpha$ obeying \eqref{eq:A}; $S$ is its set of minimizers, $m = \dim S$,
$k = n-m$, and $\pi$ is as in \eqref{eq:pi}.

\begin{thm}[Single minimizer; $\bcr$ Theorem 1.1]\label{thm:one}
If $f$ has a unique critical point $x^*$, then there is a diffeomorphism
$\phi\colon\M\to\R^n$ such that $f(x) = f(x^*) + \|\phi(x)\|^2$ for all $x$. In
particular $\M$ is diffeomorphic to $\R^n$.
\end{thm}

\begin{proof}
By Lemma~\ref{lem:traj}, sublevel sets of $f$ lie in closed balls about $x^*$,
which are compact, so $f$ is coercive; by Lemma~\ref{lem:MB} the Hessian at
$x^*$ is positive definite. Apply Theorem~\ref{thm:BCR33}.
\end{proof}

\begin{thm}[Fiber bundle with control on $f$; $\bcr$ Theorem 4.6]\label{thm:bundle}
Fix $\bar x\in S$ and let $F=\pi^{-1}(\bar x)$. Let $U\subseteq S$ be a
contractible open neighborhood of $\bar x$. There exists a map
$\phi\colon\pi^{-1}(U)\to F$ such that
$\psi := (\pi,\phi)\colon\pi^{-1}(U)\to U\times F$ is a diffeomorphism and
$f(y) = f(\phi(y))$ for all $y\in\pi^{-1}(U)$. Thus $\pi$ is a smooth fiber
bundle, trivial if $S$ (equivalently $\M$) is contractible.
\end{thm}

\begin{proof}
The proof of $\bcr$ Theorem 4.6 constructs $\phi(y)=\gamma(1)$ from the
horizontal-lift ODE and uses the relation between $\pi$ and $f$ in exactly one
place: to show that $\gamma$ does not escape on $[0,1]$. There, $f$ is constant
along $\gamma$, say equal to $\bar f$, and one needs a bound on
$\dist(\gamma(t), c(t)) = \dist(\gamma(t), \pi(\gamma(t)))$. Lemma~\ref{lem:traj}
supplies it directly:
\[
  \dist\bigl(\gamma(t),\bar x\bigr)
  \ \le\ \dist\bigl(\gamma(t),\pi(\gamma(t))\bigr) + \dist\bigl(c(t),c(1)\bigr)
  \ \le\ \Psi\bigl(\bar f - f^*\bigr) + \ell ,
\]
with $\ell$ the length of $c$. Everything else is unchanged. Note that what is
needed here is the pointwise estimate $\dist(x,\pi(x)) \le \Psi(h(x))$, which is
slightly stronger than the growth estimate \eqref{eq:growth}; this is why
Lemma~\ref{lem:traj} is stated in that form.
\end{proof}

\begin{thm}[Local and global normal form; $\bcr$ Theorems 4.7 and 1.2]\label{thm:main}
The set $S$ is a connected, properly embedded smooth submanifold of $\M$. For
every contractible open $U\subseteq S$ there is a diffeomorphism
$\psi\colon\pi^{-1}(U)\to U\times\R^k$ of the form $\psi=(\pi,\phi)$ with
\[
  f(y) = f^* + \|\phi(y)\|^2 \qquad \text{for all } y\in\pi^{-1}(U) .
\]
If $\M$ is contractible, this holds with $U=S$ and $\pi^{-1}(U)=\M$: there is a
diffeomorphism $\psi=(\pi,\phi)\colon\M\to S\times\R^k$ with
$f = f^* + \|\phi\|^2$, so that $f$ is a nonlinear least-squares function with
$\phi$ a submersion; and if $f$ is not constant then $\M$ is diffeomorphic to
$\R^n$.
\end{thm}

\begin{proof}
Combine Propositions~\ref{prop:pi}, \ref{prop:fiber} and
Theorem~\ref{thm:bundle} exactly as in $\bcr$ Section 4.3, and invoke $\bcr$
Theorem 1.3 for the last claim.
\end{proof}

The corollaries of $\bcr$ Section 1.5 are all deduced from Theorem~\ref{thm:main}
and from purely topological input, so they carry over without further comment.
We list them for completeness, assuming $\M$ contractible where $\bcr$ does.

\begin{cor}[$\bcr$ Corollary 4.5]
$\pi\colon\M\to S$ is a smooth fiber bundle with fibers diffeomorphic to $\R^k$.
\end{cor}

\begin{cor}[Compact implies point; $\bcr$ Corollary 1.4]
If $\M$ is contractible and $S$ is compact, then $S$ is a singleton, and
Theorem~\ref{thm:one} applies.
\end{cor}

\begin{cor}[$\bcr$ Corollaries 1.8 and 4.8]
If $U\subseteq S$ is open and diffeomorphic to $\R^m$, there is a
diffeomorphism $\xi\colon\pi^{-1}(U)\to\R^n$ with
$f(\xi^{-1}(y)) = f^* + y_{m+1}^2+\cdots+y_n^2$. In particular, if $\M$ is
contractible, then $f$ can be deformed into a pure quadratic if and only if $S$
is diffeomorphic to $\R^m$.
\end{cor}

\begin{cor}[$\bcr$ Corollary 1.9]
If $\M$ is contractible, the lifted function $g(x,t) = f(x)$ on $\M\times\R$ can
be deformed into a pure quadratic.
\end{cor}

\begin{thm}[Hidden convexity; $\bcr$ Theorem 1.10]\label{thm:gconvex}
If $\M$ is contractible, it admits a complete Riemannian metric
$\langle\cdot,\cdot\rangle_2$ with respect to which $f$ is geodesically convex
and globally $1$-P\L{}. If (and only if) $S$ is diffeomorphic to $\R^m$, the
metric may be chosen so that $\M$ is isometric to Euclidean space.%
\footnote{The clause ``globally $1$-P\L{}'' is already in the conclusion of $\bcr$ Theorem
1.10(a) (\cite{BCR2026}, \S1.5.3); it is easy to read past, because the section
heading mentions only the geodesic convexity, and the clause itself appears only
in passing in the $\bcr$ \S6 proof, where it follows from
$\nabla(f\circ\psi^{-1})(w,z) = (0,z)$. What is \emph{not} in \cite{BCR2026} is
any statement about the enlarged class of this note, which that paper has no
occasion to consider. Remark~\ref{rem:collapse} below is therefore a consequence
of the theorem above, obtained by feeding it the weaker hypothesis, and not a
claim to be found in \cite{BCR2026}. The nearest that paper comes to the same
thought is the opening of its \S5, where the authors note that in $\bcr$ Theorem
1.5 the metric is given and cannot be altered, ``which would make the problem
substantially easier''; that is an aside about their construction, not about
function classes.}
\end{thm}

\begin{proof}
Verbatim $\bcr$ Section 6: the proof needs only the diffeomorphism
$\psi=(\pi,\phi)$ of Theorem~\ref{thm:main} and the fact that
$f\circ\psi^{-1}(w,z) = f^* + \tfrac12\|z\|^2$ after cosmetic rescaling, then
pulls back the product metric.
\end{proof}

Theorem~\ref{thm:gconvex} deserves a comment, because it explains in one line
why the generalization was bound to be free. It is tempting to summarize the
situation by saying that our class ``collapses'' to the P\L{} class under a
change of metric, but that phrasing is loose in two ways, so let us state
exactly what is meant.

\begin{rem}[Equality of classes, after quantifying over metrics]\label{rem:collapse}
Let $\M$ be a contractible smooth manifold admitting a complete Riemannian
metric, and let $f\colon\M\to\R$ be smooth and bounded below, with
$f^*=\inf f$. Then
\[
  \left\{\begin{array}{c}
    f \text{ satisfies } \eqref{eq:GD}+\eqref{eq:A} \text{ for some}\\
    \text{positive definite } \alpha \text{ and some complete metric}
  \end{array}\right\}
  \;=\;
  \left\{\begin{array}{c}
    f \text{ is globally } 1\text{-P\L{} with respect}\\
    \text{to some complete metric}
  \end{array}\right\} .
\]
Indeed, ``$\supseteq$'' holds because global $\mu$-P\L{} is \eqref{eq:GD} with
$\alpha(s)=\sqrt{2\mu s}$, and ``$\subseteq$'' is Theorem~\ref{thm:gconvex}.
Two points are essential here, and the word ``collapses'' hides both of them.

First, the metric $\langle\cdot,\cdot\rangle_2$ produced by
Theorem~\ref{thm:gconvex} is built from $\psi$, hence depends on $f$. So this is
\emph{not} an assertion that the two classes agree relative to any one fixed
metric (relative to a fixed metric the inclusion is strict, by
Examples~\ref{ex:bounded} and \ref{ex:log}) and the distinction is exactly the
one \cite{BCR2026} draws attention to at the start of its \S5, where the metric
in $\bcr$ Theorem 1.5 is given in advance and may not be altered.

Second, contractibility of $\M$ is needed, since without it
Theorem~\ref{thm:main} yields only local trivializations and
Theorem~\ref{thm:gconvex} is unavailable. (See $\bcr$ Example 7.1 for a
globally P\L{} and g-convex function on a cylinder, showing that part (a) of
$\bcr$ Theorem 1.10 is in any case not an equivalence.)
\end{rem}

With Remark~\ref{rem:collapse} in hand, the extra generality is invisible to any
conclusion that is invariant under a change of complete metric, which is to
say, to every conclusion of \cite{BCR2026} except the metric ones. Let us add
that this phenomenon is not an artifact of the present enlargement. Already
within the class of \cite{BCR2026}, passing to $\langle\cdot,\cdot\rangle_2$
normalizes $\mu$ to $1$ and supplies geodesic convexity for free, so
``P\L{} with respect to some complete metric'' was always a much coarser
invariant than ``P\L{} with respect to a given metric''. What we have done is to
widen a known equivalence class, not to exhibit a new one.

\subsection{The achievable geometry does not grow}

There is a converse side to the story that we find worth emphasizing. One might
expect that enlarging the function class enlarges the family of minimizer sets
that can be realized. It does not.

\begin{cor}[Characterization of $S$; $\bcr$ Corollary 1.6]\label{cor:charS}
Let $\tilde S$ be a smooth manifold, fix $n > \dim\tilde S$, and endow
$\M=\R^n$ with a complete Riemannian metric. The following are equivalent:
\begin{enumerate}
\item[(a)] $\tilde S$ is diffeomorphic to the minimizer set of a smooth
$f\colon\M\to\R$ satisfying \eqref{eq:GD} for some positive definite $\alpha$
obeying \eqref{eq:A};
\item[(a$'$)] $\tilde S$ is diffeomorphic to the minimizer set of a smooth,
globally P\L{} $f\colon\M\to\R$;
\item[(b)] $\tilde S$ is contractible.
\end{enumerate}
\end{cor}

\begin{proof}
(a) $\Rightarrow$ (b) by Proposition~\ref{prop:pi}. (b) $\Rightarrow$ (a$'$) is
$\bcr$ Corollary 1.6, whose proof invokes $\bcr$ Theorem 1.5 to produce a
genuinely globally $1$-P\L{} function. And (a$'$) $\Rightarrow$ (a) because
global $\mu$-P\L{} is \eqref{eq:GD} with $\alpha(s)=\sqrt{2\mu s}$, which is
positive definite and satisfies \eqref{eq:A}.
\end{proof}

\begin{cor}[$\bcr$ Corollary 1.7]
For every $m\ge3$ and $n\ge m+1$ there is a smooth $f$ on $\R^n$ in our class
whose minimizer set is an $m$-dimensional submanifold not homeomorphic to
$\R^m$. For $m=3$ the Whitehead manifold is the classical example; for general
$m\ge3$ one takes any contractible smooth $m$-manifold not homeomorphic to
$\R^m$, as in $\bcr$.
\end{cor}

The same remark applies to embeddings. By Theorem~\ref{thm:main} and $\bcr$
Theorem 1.5, a properly embedded $S\subseteq\M$ (with $\M$ contractible) arises
as the minimizer set of some $f$ in our class if and only if there is a
diffeomorphism $\psi\colon\M\to S\times\R^k$ with
$\psi(S)=S\times\{0\}$, which is the
same criterion as for global P\L{}. In particular the necessary condition
$\pi_1(\M\setminus S)\cong\pi_1(S^{k-1})$ of $\bcr$ Appendix G is unchanged, and
nontrivial long knots in $\R^3$ remain excluded.

\subsection{Summary of which results survive}

Table~\ref{tab:status} goes through the results of \cite{BCR2026} one at a time.
Reading it, three groups are worth distinguishing. Most entries are structural
conclusions that hold verbatim, with a pointer to the corresponding statement
above. A second group is unchanged for a different reason, namely that those
results never used the P\L{} hypothesis in the first place, as observed in
Section~\ref{sec:free}; the coercive normal form of Theorem~\ref{thm:BCR33} and
the topological input are of that kind. The last group, set off below the rule, is
what does not survive: the global quadratic growth estimate, and with it the
quantitative control on $\psi$ away from $S$.

\begin{table}[hbtp]
\centering
\small
\begin{tabular}{lll}
\hline
Result of \cite{BCR2026} & Content & Status under \eqref{eq:GD}+\eqref{eq:A}\\
\hline
Lemma 2.1 (length) & bounded trajectories, limit in $S$ & holds, Lemma~\ref{lem:traj}\\
Lemma 2.1 (QG) & $h \ge \tfrac\mu2\dist(\cdot,S)^2$ globally & \textbf{fails}; local only, \eqref{eq:QGloc}\\
Lemma 2.2 & Morse--Bott at $S$ & holds, Lemma~\ref{lem:MB}\\
Prop.\ 4.2 & $\pi$ continuous, retraction, $S$ connected & holds, Prop.~\ref{prop:pi}\\
Prop.\ 4.3 & $S$ smooth, $\pi$ smooth submersion & holds, Prop.~\ref{prop:pi}\\
Prop.\ 4.4 & fibers $\cong\R^k$, $f|_F$ in the class & holds, Prop.~\ref{prop:fiber}\\
Cor.\ 4.5 & $\pi$ is a smooth fiber bundle & holds\\
Thm.\ 4.6 & trivialization compatible with $f$ & holds, Thm.~\ref{thm:bundle}\\
Thm.\ 4.7, 1.2 & $f = f^*+\|\phi\|^2$ & holds, Thm.~\ref{thm:main}\\
Thm.\ 1.1 & unique minimizer case & holds, Thm.~\ref{thm:one}\\
Thm.\ 3.3, \S3, App.\ A--D & coercive normal form, technical lemmas & unchanged (no P\L{} used)\\
Thm.\ 1.5 & construction of P\L{} functions & unchanged\\
Thm.\ 1.3, App.\ E, F & topology of contractible manifolds & unchanged\\
Cor.\ 1.4 & compact $S$ is a point & holds\\
Cor.\ 1.6, 1.7 & characterization of $S$ & holds, and unchanged, Cor.~\ref{cor:charS}\\
Cor.\ 1.8, 1.9, 4.8 & deformation to a pure quadratic & holds\\
Thm.\ 1.10 & g-convex in some complete metric & holds, Thm.~\ref{thm:gconvex}\\
\S7, Examples 7.1--7.4 & contractibility cannot be relaxed & unchanged\\
App.\ G & obstructions on the embedding $S\subseteq\M$ & unchanged\\
\hline
\S8, bullet 4 & quantitative control on $\psi$ away from $S$ & \textbf{weakens}; see \S\ref{sec:fails}\\
\hline
\end{tabular}
\caption{Which results of \cite{BCR2026} survive the replacement of global
P\L{} by $\alpha$-gradient domination with $\alpha$ positive definite and
satisfying \eqref{eq:A}, or equivalently, by $sgl$-P\L{}I, by
Proposition~\ref{prop:sgl}.}
\label{tab:status}
\end{table}

\subsection{A second proof, by reparametrization}
\label{sec:shortcut}

The route just taken redoes three steps of \cite{BCR2026} with $\Psi$ in place of
quadratic growth. There is a shorter route to the same structural conclusions,
which avoids that work altogether: manufacture from $f$ a companion function for
which \eqref{eq:PL} holds without qualification, apply \cite{BCR2026} to it
unchanged, and
transport the normal form back. What makes this possible is that a witness may
always be shrunk, since \eqref{eq:GD} is a lower bound; the freedom
afforded by 
Remark~\ref{rem:sharpest} turns out to be exactly what is needed.

\begin{prop}[Reparametrization]\label{prop:shortcut}
Let $f\colon\M\to\R$ be smooth, bounded below and $sgl$-P\L{}I, and not constant.
Then there are a constant $c>0$ and a diffeomorphism $\theta$ of $[0,\infty)$ onto
itself, smooth and with $\theta(r)=r$ for all small $r\ge0$, such that
\[
  g \ :=\ \theta(f-f^*)
\]
is smooth and globally $\tfrac{c^2}{2}$-P\L{} on $\M$. Moreover $f$ and $g$ have
the same critical set, the same negative gradient orbits and the same end-point
map, and $g=\|\phi\|^2$ implies $f = f^*+\|\tilde\phi\|^2$ with
$\tilde\phi=\Lambda\circ\phi$ for a diffeomorphism $\Lambda$ of $\R^k$ which is the
identity near the origin.
\end{prop}

\begin{proof}
\emph{Shrinking the witness.} Let $\alpha_0$ be a positive definite witness
satisfying \eqref{eq:A}, so that $\sqrt{2\mu s}\le\alpha_0(s)$ for
$0\le s\le h_0$; put $c:=\sqrt{2\mu}$ and $\varepsilon := h_0/2$. On
$[\varepsilon,\infty)$ the function $\alpha_0$ is continuous and strictly
positive, so it admits a smooth positive minorant there, by the standard
partition-of-unity construction; extend that minorant arbitrarily to a smooth
positive function $\beta$ on all of $(0,\infty)$, which is harmless because the
cutoff below annihilates it near the origin. Let $\eta$ be smooth with $\eta=1$ on
$[0,\varepsilon]$, $\eta=0$ on $[2\varepsilon,\infty)$ and $0\le\eta\le1$, and set
\[
  \hat\alpha(s) \ :=\ \eta(s)\,c\sqrt{s} \ +\ \bigl(1-\eta(s)\bigr)\beta(s) .
\]
Being a convex combination of two functions each $\le\alpha_0$ on the relevant
range, $\hat\alpha\le\alpha_0$, so $\hat\alpha$ is again a witness; it is positive
definite, smooth on $(0,\infty)$, and equal to $c\sqrt s$ on $[0,\varepsilon]$.

\emph{The reparametrization.} Put
$\hat\Psi(r) := \int_0^r ds/\hat\alpha(s)$ and
$\theta(r) := \tfrac{c^2}{4}\hat\Psi(r)^2$. For $0\le r\le\varepsilon$ we get
$\hat\Psi(r)=2\sqrt r/c$ and therefore
\[
  \theta(r) \ =\ r \qquad (0\le r\le\varepsilon) ,
\]
so $\theta$ is smooth at the origin with $\theta'(0)=1$; on $(0,\infty)$ it is
smooth because $\hat\alpha$ is smooth and positive there, and
$\theta'(r) = \tfrac{c^2}{2}\hat\Psi(r)/\hat\alpha(r)>0$. Hence $g=\theta(h)$ is
smooth on $\M$, and
\[
  \|\nabla g\| \ =\ \theta'(h)\|\nabla f\| \ \ge\ \theta'(h)\hat\alpha(h)
  \ =\ \tfrac{c^2}{2}\hat\Psi(h) ,
  \qquad\text{so}\qquad
  \|\nabla g\|^2 \ \ge\ c^2\,\theta(h) \ =\ c^2 g .
\]
As $\inf g = 0$, this is the global P\L{} inequality \eqref{eq:PL} with parameter
$c^2/2$. Since $\theta'>0$ and $\theta(0)=0$, the functions $f$ and $g$ have the
same critical set and the same minimizer set, and $\nabla g = \theta'(h)\nabla f$
is a positive multiple of $\nabla f$ depending on $x$ only through $h$. So
Lemma~\ref{lem:reparam} applies with $\sigma=\theta'$: the two negative gradient
flows have the same oriented orbits and the same end-point map. (That lemma is
elementary and uses neither \eqref{eq:GD} nor \eqref{eq:A}, so this proof still
stands on its own, without the trajectory lemma of
Section~\ref{sec:preliminaries}. By \cite{BCR2026} the $g$-trajectories are
defined for all time and converge to points of $S$, so the same is true of the
$f$-trajectories. Note also that on $\{h<\varepsilon\}$ we have $\theta'=1$, so
there the two vector fields are not merely proportional but \emph{equal}.)

\emph{Surjectivity of $\theta$, and the transfer.} Apply \cite{BCR2026} to $g$:
by $\bcr$ Proposition 4.4, each fiber $F$ of the end-point map of $g$ is a
complete, connected submanifold diffeomorphic to $\R^k$ on which $g$ restricts to
a globally P\L{} function with a single critical point, so $\bcr$ Theorem 1.1
applies to $g|_F$ and puts it in the normal form $\|\phi_F\|^2$ with
$\phi_F\colon F\to\R^k$ a diffeomorphism. Since $f$ is not constant the fibers
have dimension $k\ge1$, so $g$ already takes all values in $[0,\infty)$ on a
single fiber. Hence $g$ is unbounded. This is what forces $\hat\Psi$ to diverge,
which does not follow from its definition alone: a witness growing fast enough at
infinity would make $\int^\infty ds/\hat\alpha$ converge. Since
$\hat\alpha = c\sqrt s$ near the origin, $\hat\Psi(r)$ is finite for every finite
$r$, so $\theta$ is finite and nondecreasing on all of $[0,\infty)$; were it
bounded, $g=\theta(h)$ would be bounded too. So $\theta$ is unbounded, hence so is
$\hat\Psi$, and being nondecreasing, $\hat\Psi(r)\to\infty$ as $r\to\infty$.
Therefore $\theta$ maps $[0,\infty)$ onto itself; being smooth
with positive derivative, it is a diffeomorphism. (Remark~\ref{rem:L} is the same
fact, proved there within the development of Section~\ref{sec:preliminaries};
obtaining it here from \cite{BCR2026} applied to $g$ is what keeps the present
proof independent of that development.) Consider then the radial map
\[
  \Lambda(z) \ :=\ \sqrt{\frac{\theta^{-1}(\|z\|^2)}{\|z\|^2}}\;\,z ,
  \qquad \Lambda(0) \ :=\ 0 .
\]
The quotient under the root is identically $1$ for $\|z\|^2\le\varepsilon$, because
$\theta$ is the identity there; so $\Lambda$ is the identity near the origin, and in
particular there is nothing to check about its smoothness at that point. Away from
the origin $\Lambda$ is smooth because $\theta^{-1}$ is, and it is a bijection of
$\R^k$ because it fixes each ray and acts on the radius by the increasing bijection
$\|z\|\mapsto\sqrt{\theta^{-1}(\|z\|^2)}$ of $[0,\infty)$; hence $\Lambda$ is a
diffeomorphism of $\R^k$, and $\|\Lambda(z)\|^2 = \theta^{-1}(\|z\|^2)$ by
construction. Explicitly, the inverse is
\[
  \Lambda^{-1}(w) \ =\ \sqrt{\frac{\theta(\|w\|^2)}{\|w\|^2}}\;\,w ,
  \qquad \Lambda^{-1}(0) \ :=\ 0 ,
\]
which is smooth for the same two reasons. So $g=\|\phi\|^2$ gives
$h = \theta^{-1}(\|\phi\|^2) = \|\Lambda\circ\phi\|^2$, as claimed.
\end{proof}

Proposition~\ref{prop:shortcut} settles all of the structural conclusions at one
stroke: Theorems~\ref{thm:one}, \ref{thm:bundle} and \ref{thm:main}, and every
corollary drawn from them, follow by applying the corresponding result of
\cite{BCR2026} to $g$ and transporting back. In particular, apart from the
elementary Lemma~\ref{lem:reparam}, which is independent of \eqref{eq:GD} and
\eqref{eq:A}, Sections~\ref{sec:free}--\ref{sec:preliminaries} are not logically
necessary for them.

We have nonetheless kept the direct route, for two reasons.

The first is quantitative. It is not that the reparametrization yields no
estimates: applying the ordinary P\L{} trajectory bound to $g$, whose parameter is
$c^2/2$, gives
\[
  \operatorname{length}\bigl(\text{orbit from } x\bigr) \ \le\ \hat\Psi(h(x)) ,
  \qquad \dist(x,S) \ \le\ \hat\Psi(h(x)) ,
\]
and since $f$ and $g$ have the same orbits these are bounds for the original
gradient flow too. But they are expressed in the \emph{shrunken} witness, and
$\hat\alpha\le\alpha_0$ forces $\hat\Psi\ge\Psi_0$, so they are weaker than what
Lemma~\ref{lem:traj} gives. The direct route is what preserves the estimates in
terms of the original witness, and those are what
Section~\ref{sec:applications} uses, for instance in the tube estimate for the
scalar LQR problem, where the sharp $\Psi$ is the point.

The second is that the direct route shows which ingredient of \cite{BCR2026} each
hypothesis is answering, which was the problem that motivated this note.

\begin{rem}[The reparametrization is an equivalence]\label{rem:equiv}
Proposition~\ref{prop:shortcut} has a converse, and it costs three lines. Suppose
$g=\theta(h)$ is globally $\mu$-P\L{} for some smooth diffeomorphism $\theta$ of
$[0,\infty)$ onto itself with $\theta(0)=0$ and $\theta'>0$. From $\nabla g = \theta'(h)\nabla f$ and
$\|\nabla g\|^2\ge2\mu g$ we get $\|\nabla f\|\ge\beta(h)$ with
\[
  \beta(s) \ :=\ \frac{\sqrt{2\mu\,\theta(s)}}{\theta'(s)} ,
\]
so $\beta$ is a witness for $f$ in the sense of \eqref{eq:GD}. It is continuous and
vanishes only at $s=0$, hence positive definite; and
\[
  \frac{\beta(s)^2}{s} \ =\ \frac{2\mu\,\theta(s)}{s\,\theta'(s)^2} ,
\]
whose limit as $s\downarrow0$ is $2\mu/\theta'(0)>0$, since $\theta(s)/s\to\theta'(0)$.
So $\beta$ satisfies \eqref{eq:A} on a small enough interval. By
Proposition~\ref{prop:sgl}, $f$ is $sgl$-P\L{}I. Note that nothing here restricts the
image of $h$: the estimate is asked only at values that $h$ attains, while the
positive definiteness of $\beta$ and \eqref{eq:A} are properties of $\theta$ on
$[0,\infty)$. Combining with Proposition~\ref{prop:shortcut}:
\emph{a smooth, nonconstant $f$, bounded below, is $sgl$-P\L{}I if and only if there
is a smooth diffeomorphism $\theta$ of $[0,\infty)$ onto itself, with $\theta(0)=0$
and $\theta'>0$, such that $\theta(f-f^*)$ is globally P\L{}.} Neither direction needs
a hypothesis on $\operatorname{im}(h)$; on the contrary, $\operatorname{im}(h) =
[0,\infty)$ is a \emph{consequence}. Indeed, in the forward direction $g=\theta(h)$
is globally P\L{}, so by $\bcr$ Proposition 4.4 and Theorem 1.1 the restriction of
$g$ to any fiber
$F$ of its end-point map is $\|\phi_F\|^2$ for a diffeomorphism
$\phi_F\colon F\to\R^k$; as $f$, and hence $g$, is nonconstant we have $k\ge1$, so
already $g(F)=[0,\infty)$. Therefore $g(\M)=[0,\infty)$, and since $\theta$ is a
bijection of $[0,\infty)$ onto itself,
$h(\M)=\theta^{-1}(g(\M))=[0,\infty)$. (Only the fiberwise statement is used here,
not the global product normal form, which would need contractibility.)

Where the domain of $\theta$ does matter is in the requirement that it be all of
$[0,\infty)$. Weakening that to $\operatorname{im}(h)$ gives a strictly weaker notion,
as Example~\ref{ex:tanh} shows: there $\theta=\operatorname{arctanh}$ reparametrizes
$f$ into a globally P\L{} function, but it is a diffeomorphism $[0,1)\to[0,\infty)$
rather than of $[0,\infty)$ onto itself, and $f$ is only $loc$-P\L{}I.
\end{rem}

\begin{rem}[Why the minimal construction may require shrinking the witness]\label{rem:shrink}
The shrinking step is not a convenience. For a \emph{given} witness the function
$\theta=\tfrac\mu2\Psi^2$ can fail to be differentiable at the origin, so the
argument really does need the freedom to replace $\alpha_0$. Take $f(x)=x^2$ on
$\R$, so $h=x^2$ and $\|\nabla f\|=2\sqrt h$, and let
\[
  \alpha(s) \ =\ \sqrt s\,\bigl(1+\epsilon\sin(\log s)\bigr) ,
  \qquad 0<\epsilon<\tfrac12 ,
\]
with $\alpha(0)=0$. This is a witness, since $1+\epsilon\sin(\log s)<2$ gives
$\alpha(s)\le2\sqrt s=\|\nabla f\|$ at level $s$; it is positive definite, and
$\alpha(s)^2\ge s(1-\epsilon)^2$ gives \eqref{eq:A} with $2\mu=(1-\epsilon)^2$.
Substituting $s=rt$ gives
$\Psi(r)/\sqrt r = F(\log r)$ with
$F(\tau) = \int_0^1 t^{-1/2}\bigl(1+\epsilon\sin(\tau+\log t)\bigr)^{-1}dt$, a
nonconstant function of period $2\pi$; so $\Psi(r)^2/r$ oscillates and has no
limit as $r\downarrow0$, and $\theta'(0)$ does not exist.

The obstruction is attached to the witness and not to $f$, and it is worth saying so,
since the example is easily misread. Indeed $f(x)=x^2$ is itself globally $2$-P\L{}, so
$\theta=\operatorname{id}$ already serves. Even keeping the witness $\alpha$ above, the
linear choice $\theta(r)=\lambda r$ is admissible as soon as
$\lambda\ge2/(1-\epsilon)^2$, since then
$\theta'(s)^2\alpha(s)^2/(2\theta(s)) = \lambda\alpha(s)^2/(2s)
\ge \lambda(1-\epsilon)^2/2 \ge 1$, which is global P\L{} with constant $1$. What
fails is only the passage from this witness to the \emph{minimal} $\theta$, namely
$\tfrac12\Psi^2$, and it is that passage which the shrinking step repairs in general,
when no better witness is in hand.
\end{rem}

\begin{rem}[Every admissible witness is of exact square-root order]\label{rem:sqrtorder}
It is worth mentioning what \emph{cannot} go wrong, since it is easy to look for the
wrong obstruction. One might hope to defeat the shortcut with a witness that is too
large at the origin, say $\alpha(s)=s^{1/4}$, for which \eqref{eq:A} holds near $0$
while $\Psi(r)=\tfrac43r^{3/4}$ and $\tfrac\mu2\Psi^2$ has vanishing derivative at
the origin. No such witness exists for any \emph{nonconstant} $f$ in our class.
Indeed, by
Lemma~\ref{lem:MB} the Hessian is nondegenerate in the normal directions along
$S$, so along a unit normal geodesic $t\mapsto\Exp_{x^*}(tv)$ at $x^*\in S$ we have
$h = at^2+O(t^3)$ and $\|\nabla f\| = bt+O(t^2)$ with $a,b>0$; solving $h=s$ gives
$t\asymp\sqrt s$ and hence $\alpha_f(s) = O(\sqrt s)$. Since every witness is dominated by
$\alpha_f$, \emph{no} witness can be asymptotically larger than order $\sqrt s$, which
is what disqualifies $s^{1/4}$; and \eqref{eq:A} prevents an admissible witness from
being asymptotically smaller, so every witness satisfying \eqref{eq:A} obeys
$\alpha(s)=\Theta(\sqrt s)$ at the origin. The restriction to admissible witnesses in
that last clause is needed: for $f(x)=x^2$ the positive definite $\alpha(s)=s/(1+s)$
is a witness, since $s/(1+s)\le2\sqrt s$, and it is $o(\sqrt s)$; it simply fails
\eqref{eq:A}. (The qualifier is needed: for constant $f$
one has $h\equiv0$, so \eqref{eq:GD} is tested only at $s=0$ and every positive
definite $\alpha$ is a witness, $s^{1/4}$ included. Proposition~\ref{prop:shortcut}
excludes that case anyway.)

Exact square-root order does not by itself make the reparametrization of a given
witness smooth, however, and it would be a mistake to think that oscillation is the
only obstruction. The ratio $\alpha(s)/\sqrt s$ may also approach a limit with
insufficient regularity. For $f(x)=x^2$ again, whose sharp witness is
$\alpha_f(s)=2\sqrt s$, take
\[
  \alpha(s) \ =\
  \begin{cases}
    \sqrt s\,\bigl(1+s^{1/4}\bigr) , & 0\le s\le 1 ,\\[1mm]
    2\sqrt s , & s\ge1 .
  \end{cases}
\]
This is continuous (both branches give $2$ at $s=1$), positive definite, and
satisfies $\alpha\le\alpha_f$ on every level, so it is a witness on all of
$[0,\infty)$; it satisfies \eqref{eq:A} with $2\mu=1$. The second branch matters
only for the definition: all the asymptotics below come from the first. Indeed,
substituting $s=u^4$ there gives the closed form
for $0\le r\le1$,
$\Psi(r) = 4\bigl[r^{1/4}-\log(1+r^{1/4})\bigr] = 2r^{1/2}-\tfrac43r^{3/4}+r+O(r^{5/4})$
(beyond $r=1$ the witness is $2\sqrt s$ and $\Psi(r)=4(1-\log2)+\sqrt r-1$, which
plays no role, since only the behaviour at the origin is at issue),
so
\[
  \theta(r) \ =\ \tfrac14\Psi(r)^2 \ =\ r - \tfrac43 r^{5/4} + O(r^{3/2}) ,
\]
which is $C^1$ but not $C^2$ at the origin: $\theta''(r)$ grows like $r^{-3/4}$.
Here $\alpha(s)/\sqrt s\to1$ monotonically, with no oscillation at all. The
shrinking argument of Proposition~\ref{prop:shortcut} disposes of both
obstructions at once, which is the reason for taking the witness to be exactly
$c\sqrt s$ near the origin rather than merely of that order.

Here too the pathology belongs to the pairing of $f$ with a poorly chosen witness
rather than to $f$: the function is again $x^2$, for which $\theta=\operatorname{id}$
serves, and even for the witness displayed above the linear $\theta(r)=2r$ is
admissible, since $\alpha(s)^2\ge s$ gives
$\theta'(s)^2\alpha(s)^2/(2\theta(s)) = \alpha(s)^2/s \ge 1$. What
neither remark should be taken to say is that a smooth reparametrization fails to
exist; what fails is the recipe $\tfrac12\Psi^2$ applied to the witness one happens to
have been handed.
\end{rem}

\begin{rem}[Which reparametrizations work, and the minimal one]\label{rem:minimal}
The $\theta$ built in the proof is far from the only one that serves. Both of the
preceding remarks were settled by exhibiting some other reparametrization for the
same witness, which raises the question of how much freedom there is; and in
examples a more convenient $\theta$ is often available. Normalize the target
P\L{} constant to $1$; the constant is in any
case free, since replacing $\theta$ by $\lambda\theta$ scales it by $\lambda$. On the
strength of a witness $\alpha$ alone, a candidate $\theta$ with $\theta(0)=0$ and
$\theta'>0$ delivers the global inequality for $g=\theta(h)$ exactly when
$\theta'\alpha\ge\sqrt{2\theta}$ throughout. Measure the room in that inequality by
the \emph{level-wise slack}
\[
  \rho_\theta(h) \ :=\ \frac{\theta'(h)^2\alpha(h)^2}{2\,\theta(h)} ,
\]
which is the P\L{} constant that the witness guarantees for $g$ at level $h$, so that
admissibility reads $\rho_\theta\ge1$. Since $\theta'/\sqrt\theta = 2(\sqrt\theta)'$
and $\Psi'=1/\alpha$, we have $(\sqrt\theta)' = \sqrt{\rho_\theta/2}\;\Psi'$, and
integrating from $0$ gives the identity
\begin{equation}\label{eq:average}
  \sqrt{\frac{\theta(h)}{\theta_{\min}(h)}}
  \ =\ \frac{\int_0^h\sqrt{\rho_\theta}\;d\Psi}{\int_0^h d\Psi} ,
  \qquad \theta_{\min} \ :=\ \tfrac12\Psi^2 .
\end{equation}
That is, $\sqrt{\theta/\theta_{\min}}$ is the $d\Psi$-weighted average of
$\sqrt{\rho_\theta}$ on $[0,h]$. Two consequences. First, $\rho_\theta\ge1$ forces
$\theta\ge\theta_{\min}$ pointwise, and equality holds throughout only for
$\theta_{\min}$ itself. So $\tfrac12\Psi^2$ is the pointwise lower envelope of the
solutions of the differential inequality, and it is the minimal \emph{admissible}
reparametrization whenever it happens to be smooth. That proviso is not idle:
Remarks~\ref{rem:shrink} and~\ref{rem:sqrtorder} exhibit witnesses for which
$\tfrac12\Psi^2$ is not differentiable, or only $C^1$, at the origin, and it is then
no reparametrization at all. The $\theta$ of Proposition~\ref{prop:shortcut} is this
minimal solution for the shrunken witness, which is chosen precisely so that it is
smooth. Second, a competitor exceeds the minimal one by the average
of its own slack, and if $\rho_\theta\to L$ while $\Psi(\infty)=\infty$ then
$\theta/\theta_{\min}\to L$. Example~\ref{ex:bounded} illustrates both.
\end{rem}

\section{What is not necessarily true in this generality}
\label{sec:fails}

\subsection{Global quadratic growth}

The global inequality $h \ge \tfrac{\mu}{2}\dist(\cdot,S)^2$ of $\bcr$ (QG) is
lost, and with it every estimate derived from it far from $S$. What remains is
\eqref{eq:growth}, $h \ge \Psi^{-1}(\dist(\cdot,S))$, which is quadratic only in
the regime $h\le h_0$ where \eqref{eq:A} is active. The degradation can be
severe: for $f(x)=\log(1+\|x\|^2)$ on $\R^n$ (Example~\ref{ex:log}) we have
$h = \log(1+\dist(x,S)^2)$, so growth is logarithmic, slower than any positive
power. Since the equivalences among P\L{}, quadratic growth, Morse--Bott, error
bound and the restricted secant inequality for $C^2$ functions \cite{RB2024a}
are local statements, they persist near $S$; but globally our class is strictly
larger than each of them.

Rates change too, of course: along the negative gradient flow
$\eta' = -\|\nabla f\|^2 \le -\alpha(\eta)^2$, so global exponential decay of the
relative loss is available only when $\alpha(s)^2\gtrsim s$ globally, which is
$gl$-P\L{}I. We do not pursue this here, since it is not what the present note is
about; the rate question is the subject of \cite{Sontag2025L4DC} and the papers
cited there, where it is answered class by class, and nothing we do here bears on
it either way.

\subsection{Quantitative control on $\psi$}

The open question in $\bcr$ Section 8 about controlling the diffeomorphism
$\psi$ is unaffected \emph{at} $S$: if $\nabla f$ is $L$-Lipschitz near $S$,
then $\nabla^2 f(x) = 2\,D\phi(x)^*\circ D\phi(x)$ together with \eqref{eq:MB}
still gives $m$ singular values of $D\psi(x)$ equal to $1$ and $k$ of them in
$[\sqrt{\mu/2},\sqrt{L/2}]$. Away from $S$, however, any bound one might hope to
propagate along the flow now depends on $\alpha$ and not on $\mu$ alone, so the
question becomes strictly harder in our setting.

\section{Sharpness: the two immediate weakenings both fail}
\label{sec:examples}

By Proposition~\ref{prop:sgl} our hypothesis is $sgl$-P\L{}I, which is the
conjunction of two requirements: a $\sqrt{\,\cdot\,}$ lower bound near the origin
(that is, $loc$-P\L{}I, our \eqref{eq:A}) and a positive definite bound on every
level (that is, $\mathcal{PD}$). This section shows that neither can be dropped.

Examples \ref{ex:exp}--\ref{ex:cross} keep $\mathcal{PD}$ and drop
\eqref{eq:A}; the last of them is the sharpest, since it destroys the manifold
structure of $S$ itself. Example~\ref{ex:tanh} does the reverse: it retains the
local square-root estimate (that is, $loc$-P\L{}I) while dropping level-wise
uniform positivity. (One cannot literally say that it ``satisfies \eqref{eq:A}'',
since \eqref{eq:A} is a condition on a positive definite witness $\alpha$ and in
that example no such witness exists.) Examples \ref{ex:bounded}
and~\ref{ex:log} are not counterexamples at all: they separate $sgl$ from $sat$
and $sat$ from $gl$, and so show that the generalization is not vacuous.

We open with a remark that uses several of these examples at once, and that is
best read as a summary of what they collectively show. It asks the converse
question: the conclusion $f = f^*+\|\phi\|^2$ is weaker than the hypothesis, and
the examples below are exactly what measures the gap. A reader who wants the
examples first may skip to Example~\ref{ex:exp} and return here afterwards.

\begin{rem}[How much of the hypothesis does the conclusion remember?]
\label{rem:converse}
Theorem~\ref{thm:main} runs in one direction only, and it is worth asking what
its conclusion gives back. Suppose then that we are handed
\[
  f \ =\ f^* + \|\phi\|^2 , \qquad
  \phi\colon\M\to\R^k \ \text{ a smooth submersion} ,
\]
and nothing else. Everything below comes from a single identity. Put
\[
  G(x) \ :=\ D\phi(x)\circ D\phi(x)^* ,
\]
a positive definite operator on $\R^k$ because $\phi$ is a submersion. Then
$\nabla f = 2\,D\phi^*\phi$, so at every $x$ with $\phi(x)\ne0$,
\begin{equation}\label{eq:rayleigh}
  \|\nabla f(x)\|^2
  \ =\ 4\bigl\langle \phi(x), G(x)\phi(x)\bigr\rangle
  \ =\ 4\,h(x)\,\bigl\langle u(x), G(x)u(x)\bigr\rangle ,
  \qquad u \ :=\ \phi/\|\phi\| .
\end{equation}
Every P\L{} condition on $f$ therefore translates into a uniformity statement
about the Rayleigh quotient of $G$ in the direction of the residual $\phi$;
global $\mu$-P\L{} is exactly $\langle u, Gu\rangle \ge \mu/2$ throughout.

The qualitative half of Section~\ref{sec:preliminaries} comes for free. Since
$D\phi(x)^*$ is injective, $\nabla f(x)=0$ if and only if $\phi(x)=0$, so $f$ is
invex and $S = \phi^{-1}(0)$ is a closed, properly embedded submanifold of
codimension $k$, being the preimage of a regular value; and at $x\in S$ we have
$\nabla^2f(x) = 2\,D\phi(x)^*D\phi(x)$, whose kernel is $T_xS$ and whose
eigenvalues on $N_xS = \operatorname{im}D\phi(x)^*$ are the numbers
$2\lambda_i(G(x))>0$. So the Morse--Bott property holds pointwise.

Every trace of uniformity is lost, however, and in four separate ways.
\begin{enumerate}
\item[(i)] $S$ may be empty. Take $\phi(x) = e^{-x/2}$ on $\R$, a submersion, and
$f = e^{-x}$ is Example~\ref{ex:exp} below.
\item[(ii)] $S$ need not be connected, so the deformation retraction of
Proposition~\ref{prop:pi} is lost as well. Identify $\R^2$ with $\C$ and take
$\phi(z) = e^z-1$, whose differential is never singular, so that
$f = |e^z-1|^2$ has $S = 2\pi i\,\Z$, a discrete set.
\item[(iii)] $\inf_S \lambda_{\min}(G)$ may vanish, so \eqref{eq:MB} with a
single constant $\mu$ is strictly more than the normal form supplies.
\item[(iv)] No P\L{}I of any of the four kinds follows, not even $loc$-P\L{}I.
\end{enumerate}

For (iv) we give two examples, because the first suggests a repair that the
second defeats. Example~\ref{ex:tanh} below is of the present form:
$\tanh(x^2) = \phi(x)^2$ with $\phi(x) = x\sqrt{\tanh(x^2)/x^2}$, the second
factor extending smoothly and positively to the origin, and $\phi'>0$
everywhere, so $\phi$ is a submersion; but its image is $(-1,1)$. One might hope
that asking $\phi$ to be onto $\R^k$ repairs matters. It does not. On $\M=\R^2$
take
\[
  \phi(x,y) \ =\ y\,e^{-x^2} , \qquad
  f \ =\ \phi^2 \ =\ y^2e^{-2x^2} ,
\]
a surjective submersion with $S = \{y=0\}$, for which \eqref{eq:rayleigh} reads
\[
  \frac{\|\nabla f\|^2}{h} \ =\ 4e^{-2x^2}\bigl(1+4x^2y^2\bigr) .
\]
Along $y=\varepsilon$ with $x\to\infty$ the left side tends to $0$ while
$h\to0$, so the infimum over $\{0<h\le\rho\}$ vanishes for every $\rho>0$ and
$loc$-P\L{}I fails. (Item (iii) is visible here too, since
$\nabla^2f(x,0) = \operatorname{diag}(0, 2e^{-2x^2})$.) Yet every conclusion of
Theorem~\ref{thm:main} holds for this $f$. Reparametrizing time by
$d\tau = e^{-2x^2}dt$ turns the negative gradient flow into
$\dot y = -2y$, $\dot x = 4xy^2$, which gives $y\to0$ and $x\to xe^{y^2}$, so
$\pi(x,y) = (xe^{y^2},0)$; and $\psi = (\pi,\phi)$, that is
$(x,y)\mapsto(xe^{y^2}, ye^{-x^2})$, is a diffeomorphism onto $S\times\R$,
its Jacobian determinant being $e^{y^2-x^2}(1+4x^2y^2)>0$ and its restriction to
each fiber $\{xe^{y^2}=p\}$, parametrized by $y$ as $x = pe^{-y^2}$, carrying
that fiber increasingly onto $\{p\}\times\R$.

Two converses do hold. Quantifying over metrics restores the implication in
full: if $\psi=(\pi,\phi)$ is a diffeomorphism onto $S\times\R^k$ then pulling
back the product metric makes $f$ globally $1$-P\L{}, which is
Theorem~\ref{thm:gconvex} read backwards and is the substance of
Remark~\ref{rem:collapse}. In a fixed metric it is enough to add properness: if
$f$ is coercive then $\lambda_{\min}(G)$ is bounded below on each compact
sublevel set, so \eqref{eq:rayleigh} yields $sgl$-P\L{}I outright. That second
converse has limited reach, since coercivity makes $S$ compact and hence, on
contractible $\M$, a single point.

So Theorem~\ref{thm:main} cannot be reversed in a fixed metric for a structural
reason: the normal form is invariant under diffeomorphism while $sgl$-P\L{}I is
not, and \eqref{eq:rayleigh} locates the discrepancy exactly, in that a
submersion controls $G$ pointwise whereas the hierarchy asks for control level by
level. The identity also sharpens the sufficient condition of
\cite[\S1.1]{BCR2026}, where $f=\tfrac12\|F-b\|^2$ is shown to be globally P\L{}
as soon as the smallest singular value of $DF^*$ has positive infimum, with the
parenthetical warning that this is not necessary: by \eqref{eq:rayleigh} the
sharp quantity is the Rayleigh quotient of $G$ along the residual, which that
singular value bounds below uniformly in the direction.
\end{rem}

\begin{ex}[Positive definiteness alone does not even give a minimizer]\label{ex:exp}
Let $\M=\R$ and $f(x)=e^{-x}$, so $f^*=0$ and $h=f$. Then
$\|\nabla f\| = e^{-x} = h$, so \eqref{eq:GD} holds with equality for the
positive definite $\alpha(s)=s$. Yet $S=\varnothing$: no conclusion of
Section~\ref{sec:survive} can hold. Of course \eqref{eq:A} fails, since
$s^2 \ge 2\mu s$ is false near $0$; consistently,
$\Psi(h)=\int_0^h ds/s = \infty$, and the trajectory
$x' = e^{-x}$ runs to $+\infty$ with infinite length.
\end{ex}

\begin{ex}[Finiteness of $\Psi$ is not enough either]\label{ex:x4}
Let $\M=\R$ and $f(x)=x^4$. Then $h=x^4$ and $\|\nabla f\| = 4|x|^3 = 4h^{3/4}$,
so \eqref{eq:GD} holds with $\alpha(s)=4s^{3/4}$, which is positive definite
with $\Psi(h) = h^{1/4} < \infty$. Here $S=\{0\}$ is a manifold and
$\M\cong\R$, so nothing visible has gone wrong, but the conclusion of
Theorem~\ref{thm:one} fails. To see this, suppose $x^4=\phi(x)^2$ for a smooth
$\phi$. Then $|\phi(x)|=x^2$, so $\phi(x)=O(x^2)$ and therefore $\phi'(0)=0$;
hence $\phi$ is not a local diffeomorphism at the origin. (One cannot argue
instead that $\phi=\pm x^2$ is non-injective: the sign may switch at the origin,
and indeed $\phi(x)=x|x|$ is an injective (though not smooth enough)
solution.) The obstruction is
the degenerate Hessian at the origin, i.e.\ the failure of Lemma~\ref{lem:MB},
i.e.\ the failure of \eqref{eq:A}. This is the minimal illustration that
\eqref{eq:A} is not bookkeeping: it is the hypothesis that makes the normal form
quadratic.
\end{ex}

\begin{ex}[The minimizer set need not be a manifold]\label{ex:cross}
Let $\M=\R^2$ and $f(x,y) = x^2y^2$, a polynomial, hence $C^\infty$. Then
$f^*=0$ and $S = \{xy=0\}$ is a cross, which is not a manifold, so
Lemma~\ref{lem:MB}, Proposition~\ref{prop:pi} and Theorem~\ref{thm:main} all
fail. Nonetheless $f$ is gradient dominated by a positive definite $\alpha$ with
finite $\Psi$: from $\|\nabla f\|^2 = 4x^2y^2(x^2+y^2) = 4f\cdot(x^2+y^2)$ and
$x^2+y^2 \ge 2|xy| = 2\sqrt f$ we get
\[
  \|\nabla f(x,y)\|^2 \ \ge\ 8\, f(x,y)^{3/2} ,
  \qquad\text{i.e.}\qquad \alpha(s) = 2\sqrt2\, s^{3/4} ,
\]
positive definite, with $\Psi(h) = \sqrt2\,h^{1/4}$. Only \eqref{eq:A} fails.
This is a smooth analogue of the $C^1$ example in $\bcr$ footnote 3, where
$f(x,y)=x^2y^2/(x^2+y^2)$ is globally P\L{} with the same cross as its minimizer
set; the trade is that we gain full smoothness and lose the exponent. It seems to
us that this is the cleanest available demonstration that the \L{}ojasiewicz
exponent $\tfrac12$, and not smoothness, is what makes $S$ a manifold.
\end{ex}

\begin{ex}[$loc$-P\L{}I without $\mathcal{PD}$: the gradient may die at infinity]
\label{ex:tanh}
Now the reverse deficiency. Let $\M=\R$ and
\[
  f(x) \ =\ \tanh(x^2) ,
\]
so that $f^*=0$, attained only at $x^*=0$, and $f$ is smooth with
$f''(0)=2>0$. Near the origin $f(x) = x^2 + O(x^6)$ and $f'(x)=2x+O(x^5)$, so
$\|\nabla f\|^2 \ge 2h$ for $h$ small: $loc$-P\L{}I holds. The unique critical
point is the minimizer, so $f$ is invex.
Moreover $\alpha_f(r)>0$ for every $r$ in the range $[0,1)$ of $h$: writing
$x=\sqrt{\operatorname{arctanh} r}$,
\[
  \alpha_f(r) \ =\ 2\sqrt{\operatorname{arctanh} r}\;(1-r^2) \ >\ 0
  \qquad (0<r<1) .
\]
What fails is uniformity: $\alpha_f(r)\to0$ as $r\uparrow 1$, so no continuous
positive definite $\alpha$ can minorize $\alpha_f$ (continuity would force
$\alpha(1)=0$), and $f$ is not $sgl$-P\L{}I. Consistently, $1/\alpha_f(r)$ blows
up at $r=1$ at the order of $\bigl[(1-r)\sqrt{\log\tfrac{1}{1-r}}\bigr]^{-1}$,
which is not integrable, so $\int_0^1 dr/\alpha_f(r) = \infty$. (We write the
integral of $1/\alpha_f$ rather than $\Psi(1)$, since $\Psi$ was defined from a
continuous positive definite witness and here there is none.)

And the conclusion of Theorem~\ref{thm:one} does fail: $f$ is bounded, so
$f=\phi^2$ forces $\phi$ bounded, and no bounded $\phi$ is a diffeomorphism onto
$\R$. Of course $f$ is not coercive either, which is the same fact seen from
another angle.

It should be said at once how narrowly the conclusion fails, since less is lost here
than the previous paragraph may suggest. Take $\theta=\operatorname{arctanh}$ on the
range $[0,1)$ of $h$, which is a smooth diffeomorphism onto $[0,\infty)$ fixing the
origin with $\theta'(0)=1$; then $g=\theta(f)=x^2$ is globally $2$-P\L{}. So $f$ is a
smooth monotone reparametrization of a globally P\L{} function even though it is not
$sgl$-P\L{}I, and $f=\Phi^2$ with
$\Phi(x)=x\sqrt{\tanh(x^2)/x^2}$ a submersion, the quotient extending smoothly and
positively to the origin. What the reparametrization cannot do is make $\theta$ a
diffeomorphism of $[0,\infty)$ onto itself, and correspondingly $\Phi$ is not onto
$\R$: its image is $(-1,1)$. This is exactly the distinction drawn at the end of
Remark~\ref{rem:equiv}: requiring $\theta$ to be a diffeomorphism of $[0,\infty)$
onto itself is strictly stronger than requiring it only on the range of $h$, and this
example separates the two. It also locates the deficiency precisely. The nonlinear
least-squares form survives, together with the
fiber bundle structure and the topology of $S$; what fails is the identification of
the target with $\R^k$, hence the pure-quadratic normal form of
Theorem~\ref{thm:one}, and with it Theorem~\ref{thm:gconvex}, since a nonconstant
function that is globally P\L{} for a complete metric is unbounded above.

Two further points about this example should be noted. First, every individual
trajectory of $\dot x = -f'(x)$ does converge to $x^*$, with finite length
$|x_0|$. There is even an exact formula for that length in terms of $h(x_0)$,
namely $|x_0| = \sqrt{\operatorname{arctanh} h(x_0)}$. What fails is that this
function is not finite on all of $[0,\infty)$: it blows up as $h\uparrow1$, and for
$\rho\ge1$ the sublevel set $\{h\le\rho\}$ is all of $\R$, with unbounded
trajectory lengths. A finite continuous bound on $[0,\infty)$ is precisely what
Lemma~\ref{lem:traj} supplies and precisely what the escape argument of
Theorem~\ref{thm:bundle} needs. Second, $f$ \emph{does} satisfy
the compact form of $sgl$-P\L{}I discussed in Remark~\ref{rem:twosgl}: on
$[-N,N]$ the ratio $f'(x)^2/f(x)$ is continuous and positive (tending to $4$ at
the origin), hence bounded below. So the sublevel form is genuinely stronger than
the compact form, and it is the sublevel form that the structure theory needs.
\end{ex}

\begin{ex}[$sat$-P\L{}I: the motivating condition]\label{ex:bounded}
Let $\alpha(s) = \sqrt{2\mu s/(1+s)}$, that is,
\[
  \|\nabla f(x)\|^2 \ \ge\ \frac{2\mu\,h(x)}{1+h(x)} .
\]
This is exactly $sat$-P\L{}I with $a = 2\mu$ and $b=1$. This $\alpha$ is positive
definite, and it satisfies \eqref{eq:A}, though not with the same constant
$\mu$, since $2\mu s/(1+s)<2\mu s$ for every $s>0$. Because $s/(1+s)\ge s/2$ on
$[0,1]$ we get $\alpha(s)^2 \ge \mu s$ there, so \eqref{eq:A} holds with
\[
  \mu_{\eqref{eq:A}} = \tfrac{\mu}{2} , \qquad h_0 = 1 ,
\]
and more generally with any $\tilde\mu<\mu$ on a small enough interval.
Substituting
$s=\sinh^2 u$ gives the closed form
\[
  \Psi(h) = \frac{1}{\sqrt{2\mu}}\Bigl[\sqrt{h(1+h)} + \operatorname{arcsinh}\sqrt h\,\Bigr]
  = \sqrt{\tfrac{2h}{\mu}}\bigl(1+O(h)\bigr) \ \ (h\to0),
  \qquad \sim \frac{h}{\sqrt{2\mu}} \ \ (h\to\infty) ,
\]
so the local estimate of Lemma~\ref{lem:traj} recovers $\bcr$ Lemma 2.1 exactly
in the limit, while the global one degrades from quadratic to linear growth.
Here $\Psi^2$ is real-analytic and $\tfrac\mu2\Psi(h)^2 = h(1+O(h))$, so this is
one of the cases in which the reparametrization of
Proposition~\ref{prop:shortcut} needs no shrinking of the witness at all:
$g=\tfrac\mu2\Psi(h)^2$ is already smooth and genuinely globally $\mu$-P\L{}, with
the same critical set, the same orbits and, by Lemma~\ref{lem:reparam}, the same
end-point map as $f$.

A much simpler $\theta$ is available here, and it is a good illustration of
Remark~\ref{rem:minimal}. Take
\[
  \theta(t) \ =\ \frac{t(1+t)}{\mu} ,
  \qquad\text{so that}\qquad
  \frac{\|\nabla g\|^2}{2g} \ \ge\ \frac{\theta'(h)^2\alpha(h)^2}{2\theta(h)}
  \ =\ \frac{(1+2h)^2}{(1+h)^2} \ \ge\ 1 ,
\]
and $g=\theta(h)$ is globally $1$-P\L{}: a polynomial reparametrization does the work
of the transcendental one. The two are genuinely different, not rescalings of each
other. Normalizing both to the constant $1$, the minimal choice is
$\theta_{\min}=\tfrac12\Psi^2 = \tfrac{1}{4\mu}\bigl[\sqrt{h(1+h)}+\operatorname{arcsinh}\sqrt h\,\bigr]^2$, and
\[
  \theta_{\min} \ =\ \frac{h}{\mu} + \frac{h^2}{3\mu} + O(h^3) ,
  \qquad\text{against}\qquad
  \theta \ =\ \frac{h}{\mu} + \frac{h^2}{\mu} .
\]
They agree to first order at the origin, as \eqref{eq:average} requires, since the
slack $\rho_\theta(h)=(1+2h)^2/(1+h)^2$ equals $1$ there. At the other end
$\rho_\theta\to4$, and \eqref{eq:average} evaluates in closed form: for every $h>0$,
\[
  \sqrt{\frac{\theta(h)}{\theta_{\min}(h)}}
  \ =\ \frac{\int_0^h\frac{1+2s}{\sqrt{s(1+s)}}\,ds}{\sqrt{2\mu}\,\Psi(h)}
  \ =\ \frac{2\sqrt{h(1+h)}}{\sqrt{h(1+h)}+\operatorname{arcsinh}\sqrt h} ,
\]
the antiderivative being $2\sqrt{s+s^2}$ exactly. The right-hand side increases to
$2$ as $h\to\infty$, so $\theta$ exceeds the minimal reparametrization by the factor
$4$ asymptotically, which is precisely the limiting slack in its own inequality.

An equivalent and more transparent reading of this condition is worth mentioning,
as it explains why nothing was lost. Because $s/(1+s) \ge s/2$ for $s\le1$ and
$\ge 1/2$ for $s\ge1$, the inequality is equivalent (up to constants) to the
conjunction of a local P\L{} inequality on $\{h\le1\}$ and a uniform lower bound
$\|\nabla f\| \ge c > 0$ on $\{h\ge1\}$. It is, in other words, jobs (J1) and
(J2) with nothing else, which is exactly the general condition of this note,
specialized to a witness that is bounded below away from $S$.

Finally, the class is strictly larger than global P\L{}:
$f(x) = \sqrt{1+\|x\|^2}-1$ on $\R^n$ satisfies the displayed inequality with
$\mu=\tfrac12$, since $\|\nabla f\|^2 = h(h+2)/(1+h)^2$ and
$(h+2)/(1+h)\ge1$; but $\|\nabla f\| < 1$ while $h\to\infty$, so $f$ is not
globally P\L{}. Consistently with Theorem~\ref{thm:main}, $f = \|\phi\|^2$ with
$\phi(x) = x\bigl(1+\sqrt{1+\|x\|^2}\bigr)^{-1/2}$, a diffeomorphism of $\R^n$.
\end{ex}

\begin{ex}[Gradients may vanish at infinity]\label{ex:log}
Let $\M=\R^n$ and $f(x) = \log(1+\|x\|^2)$. Then $h=f$, and with $r=\|x\|$ we
have $\|\nabla f\| = 2r/(1+r^2)$ and $r^2 = e^h-1$, so
\[
  \alpha(s) = 2\sqrt{e^s-1}\;e^{-s} ,
\]
which is positive definite. Since $\alpha(s)^2/s\to4$ as $s\downarrow0$,
\eqref{eq:A} holds with any $\mu<2$ on a small enough interval, but not with
$\mu=2$ on any interval. An explicit admissible pair is
\[
  \mu = \tfrac12 , \qquad h_0 = \tfrac12 ,
\]
because $e^s-1\ge s$ and $e^{-2s}\ge e^{-1}$ for $0\le s\le\tfrac12$, so that
$\alpha(s)^2 \ge 4s/e > s = 2\mu s$ there. (The pair $\mu=1$, $h_0=\tfrac12$ would
not do: at $s=\tfrac12$ one has $\alpha(s)^2\approx0.955<1$.) So $f$ is in our
class. It is not globally P\L{},
nor does it satisfy the condition of Example~\ref{ex:bounded}, because
$\|\nabla f\|\to0$ as $h\to\infty$ while both of those conditions keep the
gradient bounded below on $\{h\ge1\}$. The conclusions nevertheless hold, as one
checks by hand: $\Psi(h) = \int_0^h e^s(e^s-1)^{-1/2}\,ds/2 = \sqrt{e^h-1}$, so
$\Psi(h(x)) = \|x\| = \dist(x,S)$ with equality in \eqref{eq:growth}, the
radial flow $\dot r = -2r/(1+r^2)$ has length exactly $r_0$, and
$f = \|\phi\|^2$ with
\[
  \phi(x) \ =\ x\,\sqrt{\frac{\log(1+\|x\|^2)}{\|x\|^2}} ,
\]
which is a diffeomorphism of $\R^n$ because $t\mapsto\log(1+t)/t$ extends to a
smooth positive function at $t=0$; written this way there is no need to except the
origin, whereas the equivalent $x\,\|x\|^{-1}\sqrt{\log(1+\|x\|^2)}$ would have to
be completed by $\phi(0)=0$.

This example is instructive twice over. It shows that our class is strictly
larger than the motivating class; and it shows that
$\dist(\cdot,S)$ can grow exponentially in $h$, so that no polynomial global
growth estimate is available.
\end{ex}

\section{What the conclusions say for LQR and logistic regression}
\label{sec:applications}

We return to the two problems of Section~\ref{sec:why}. Both satisfy
$sgl$-P\L{}I, hence the hypotheses of Section~\ref{sec:preliminaries}, so the
whole of Section~\ref{sec:survive} applies to them and it seems worth reviewing
what it actually says. (Both also satisfy a saturated inequality, under the
additional estimates recalled below; but that is more than the structural results
require, and in the scalar problem treated first we prove it rather than assume
it.) The passage from the ingredients we do use to a saturated estimate is in any
case immediate, and we may as well give it once. Suppose $\|\nabla f\|^2\ge ch$ on
$\{h\le1\}$ and $\|\nabla f\|\ge\xi(h)$ throughout for some class $\mathcal{K}$
function $\xi$. On $\{h\ge1\}$ monotonicity gives $\|\nabla f\|\ge\xi(1)>0$, while
$h/(1+h)<1$ there; on $\{h\le1\}$ we have $h/(1+h)\le h$. So
\[
  \|\nabla f\|^2 \ \ge\ \frac{a\,h}{1+h}
  \qquad\text{globally, with } a := \min\{c,\ \xi(1)^2\} ,
\]
which is $sat$-P\L{}I. In this
section we revert to the notation of \cite{Sontag2025L4DC}, reversing the
dictionary of Section~\ref{sec:setting}: $\mathcal{L}$ for the loss,
$\underline{\mathcal{L}} = \inf_D\mathcal{L}$ for its minimal value, $k$ for the
parameter being optimized, $D$ for the (open) domain, and
$h = \mathcal{L}-\underline{\mathcal{L}}$ as before. The reason for reverting is
that $x$ is now the state of the control system and $k$ the feedback matrix, so
retaining $x$ for the optimization variable would be more confusing than the
change of alphabet.

\subsection{The scalar continuous-time LQR problem}

Take the loss \eqref{eq:lqrscalar} on $D=(0,+\infty)$, with
$\underline{\mathcal{L}}=1$ attained at $k=1$. We saw in Section~\ref{sec:why}
that no $gl$-P\L{}I estimate exists, and indeed no estimate with unbounded
$\alpha$. Nevertheless:

\begin{prop}[$sat$-P\L{}I with sharp constants]\label{prop:lqrsat}
For all $k>0$,
\[
  |\mathcal{L}'(k)|^2 \ \ge\ \frac{h(k)}{1+4h(k)}
  \qquad\text{that is, } sat\text{-P\L{}I with } a=b=\tfrac14 ,
\]
and hence also $|\mathcal{L}'|^2 \ge \tfrac14\,h/(1+h)$, which is the condition of
Example~\ref{ex:bounded} with $2\mu=\tfrac14$. In both forms the constant
$a=\tfrac14$ is sharp.
\end{prop}

\begin{proof}
Here $h = (k-1)^2/(2k)$ and $\mathcal{L}'(k) = (k-1)(k+1)/(2k^2)$. For $k\ne1$ we
may divide the desired inequality by $(k-1)^2$; using
$1+4h = (k+2(k-1)^2)/k = (2k^2-3k+2)/k$, it becomes
\[
  \frac{(k+1)^2}{4k^4} \ \ge\ \frac{1}{2(2k^2-3k+2)} ,
\]
that is, $2(k+1)^2(2k^2-3k+2)\ge 4k^4$. Expanding the left side gives
$4k^4 + 2k^3 - 4k^2 + 2k + 4$, so the inequality reduces to
\[
  p(k) \ :=\ k^3-2k^2+k+2 \ \ge\ 0 \qquad (k>0) .
\]
Since $p'(k) = (3k-1)(k-1)$, the only interior local minimum on $(0,\infty)$ is at
$k=1$, where $p(1)=2>0$; and $p(0^+)=2>0$ with $p(k)\to\infty$. So $p>0$ on
$(0,\infty)$, with the deficit $2p(k)\ge4$. Sharpness: as $k\to\infty$ we have
$h\to\infty$, $|\mathcal{L}'|^2\to\tfrac14$ and $h/(1+4h)\to\tfrac14$, so no
$a>\tfrac14$ is admissible.
\end{proof}

The loss is coercive on $D$, since $\mathcal{L}(k)\to\infty$ both as $k\to0^+$
and as $k\to\infty$; so Remark~\ref{rem:coercive} applies even though $D$ is not
complete in the Euclidean metric. The unique critical point is $k=1$, with
$\mathcal{L}''(1)=1>0$. Theorem~\ref{thm:one} therefore asserts the existence of
a diffeomorphism $\phi\colon(0,\infty)\to\R$ with
$\mathcal{L}=\underline{\mathcal{L}}+\phi^2$, and in this example one can write it
down:
\[
  \mathcal{L}(k) - 1 \ =\ \frac{(k-1)^2}{2k} \ =\ \phi(k)^2 ,
  \qquad \phi(k) \ :=\ \frac{k-1}{\sqrt{2k}} ,
  \qquad \phi'(k) \ =\ \frac{k+1}{(2k)^{3/2}} \ >\ 0 ,
\]
with $\phi(k)\to\mp\infty$ as $k\to0^+,\infty$. So $\phi$ is indeed a
diffeomorphism onto $\R$, exhibiting the normal form and, incidentally, the
diffeomorphism $D\cong\R$ predicted by Theorem~\ref{thm:one}.

Two comments. First, the normal form is emphatically \emph{not} an isometry: in
the $\phi$ coordinate the loss is a perfect parabola, yet the gradient flow of
$\mathcal{L}$ is not the gradient flow of $\phi^2$, because $\phi$ distorts the
metric by the factor $\phi'$. So the normal form is a statement about the smooth
structure and leaves the metric behavior of Figure~\ref{fig:lqr}(a) exactly as it
was; nothing in either statement constrains the other. Second, $\Psi$ can be computed here from
Example~\ref{ex:bounded} with $2\mu=\tfrac14$ (the weaker of the two forms in
Proposition~\ref{prop:lqrsat} suffices), giving
\[
  |k-1| \ \le\ \Psi(h) \ =\ 2\bigl[\sqrt{h(1+h)}+\operatorname{arcsinh}\sqrt h\,\bigr] ,
\]
a bound which is $O(\sqrt h)$ near $k=1$ and $O(h)$ far away, matching the two
regimes of the figure. (Note that $\dist(k,S)=|k-1|$ here, the metric being
Euclidean on $D=(0,\infty)$. The bound is asymptotically sharp as $k\to\infty$:
there $h\sim k/2$, so both sides are $\sim k$.)

\subsection{The general continuous-time LQR problem}

Now let $\dot x = Ax+Bu$ with $(A,B)$ stabilizable, $Q\succeq0$, $R\succ0$,
$u=-kx$, and
\[
  \mathcal{L}(k) \ :=\ \operatorname{tr}(P_k)
\]
on the open set $D = D_{(A,B)} = \{k : A-Bk \text{ Hurwitz}\}\subset\R^{m\times n}$,
with $P_k$ and $Y_k$ solving the two Lyapunov equations of
\cite[\S2]{Sontag2025L4DC}; assume as there that $(A,B)$ is stabilizable, $R\succ0$,
$Q\succeq0$ and $(A,\sqrt Q)$ observable. We have written the loss as
$\operatorname{tr}(P_k)$ directly rather than as an expected cost, since
$\mathbb{E}[x_0^\top P_k x_0] = \operatorname{tr}\bigl(P_k\,\mathbb{E}[x_0x_0^\top]\bigr)$
agrees with $\operatorname{tr}(P_k)$ only when the initial covariance is the
identity; a general positive definite covariance can be accommodated, but it is a
separate matter and we do not want it left implicit. Everything we need is proved
in \cite{CuiJiangSontag2024}, in two separate places.

\begin{thm}[Cui, Jiang and Sontag \cite{CuiJiangSontag2024}]\label{thm:cjs}
\phantom{.}
\begin{enumerate}
\item[(a)] \emph{(Coercivity.)} $\mathcal{L}$ is coercive on $D$: if $k_j\to\partial D$
or $\|k_j\|_F\to\infty$ then $\mathcal{L}(k_j)\to\infty$. Equivalently, the
sublevel sets of $\mathcal{L}$ are compact subsets of $D$.
\item[(b)] \emph{(The CJS-P\L{} estimate, for ``comparison just saturated''.)}
There is a class $\mathcal{K}$ function
$\xi_1$ with $\|\nabla\mathcal{L}(k)\|_F \ge \xi_1(h(k))$ for all $k\in D$.
\end{enumerate}
\end{thm}

Part (a) is the coercivity lemma of \cite{CuiJiangSontag2024}, proved there for
self-containedness and noted to be already known; it also appears as one of the
analytical properties collected in \cite[\S3.4]{BuMesbahi2020}. Part (b) is the
gradient dominance lemma, and it improves on what was previously available, where
the P\L{} inequality was known only on compact subsets
\cite{BuMesbahi2020, Mohammadi2022}. The stronger statement that $\mathcal{L}$
satisfies $sat$-P\L{}I is asserted in \cite{Sontag2025L4DC}, and would of course
also suffice; but we do not need it, and it is instructive to see why not.

A remark on vocabulary, since the two literatures name this differently. In
\cite{CuiJiangSontag2024} the conclusion is packaged as the single statement that
$\mathcal{L}$ is a \emph{proper objective function}: by definition, one whose
gradient is locally Lipschitz, whose relative loss
$\mathcal{L}-\underline{\mathcal{L}}$ is a \emph{size function} for
$(D,\{k_{\mathrm{opt}}\})$, and which satisfies the CJS-P\L{} estimate. The middle
clause is where coercivity lives: a size function is required to be proper, that
is, to blow up at $\partial D$ and at infinity, so \emph{properness in that sense
is exactly the coercivity of Theorem~\ref{thm:cjs}(a)} and exactly the hypothesis
of Remark~\ref{rem:coercive}. Of the three clauses we use the second and the third;
the local Lipschitz condition on $\nabla\mathcal{L}$, which the first
clause supplies and which is what the ISS argument needs, we do not use at all.

\begin{prop}[CT LQR: connectedness, critical set, and $sgl$-P\L{}I]\label{prop:lqrsgl}
The critical set of $\mathcal{L}$ is the single point $k_{\mathrm{opt}}$; the
domain $D$ is connected; and $\mathcal{L}$ satisfies $sgl$-P\L{}I on $D$.
\end{prop}

\begin{proof}
\emph{The critical set.} Since $\nabla\mathcal{L}(k) = 2(Rk-B^\top P_k)Y_k$ with
$Y_k\succ0$, $\nabla\mathcal{L}(k)=0$ forces $Rk = B^\top P_k$; feeding this back
into the Lyapunov equation for $P_k$ makes $P_k$ a stabilizing solution of the
algebraic Riccati equation, and that solution is unique, so
$k=k_{\mathrm{opt}}$.

\emph{Connectedness.} This must be settled before the results of
Section~\ref{sec:survive} may be applied, since those assume a connected domain.
By Theorem~\ref{thm:cjs}(a) the restriction of $\mathcal{L}$ to any connected
component of $D$ is coercive, so that component contains a minimizer of
$\mathcal{L}$, hence a critical point. By the previous paragraph there is only
one critical point in all of $D$. So $D$ has exactly one component.

\emph{$sgl$-P\L{}I.} By Proposition~\ref{prop:sgl} it is enough to produce
$\mathcal{PD}$ and $loc$-P\L{}I, and the two come from different places.

$\mathcal{PD}$ is immediate from Theorem~\ref{thm:cjs}(b), since
$\mathcal{K}\subset\mathcal{PD}$.

For $loc$-P\L{}I we verify the hypotheses of Lemma~\ref{lem:nondeg}. Coercivity
is Theorem~\ref{thm:cjs}(a) and uniqueness of the critical point is the first
paragraph, so only nondegeneracy of the Hessian is left. Use the second-order
expansion of \cite{CuiJiangSontag2024}: for
$k,k+E\in D$,
\[
  \mathcal{L}(k+E) = \mathcal{L}(k) + 2\operatorname{tr}\bigl(E^\top(Rk-B^\top P_k)Y_k\bigr)
  + \operatorname{tr}(E^\top R E\,Y_k)
  + 2\operatorname{tr}\bigl(E^\top(Rk-B^\top P_k)\Delta Y_k\bigr) + O(\|E\|_F^3) .
\]
At $k=k_{\mathrm{opt}}$ the optimality condition $Rk = B^\top P_k$ kills both the
first-order term and the fourth term, leaving
\[
  \tfrac12\bigl\langle E, \nabla^2\mathcal{L}(k_{\mathrm{opt}})E\bigr\rangle
  = \operatorname{tr}\bigl(E^\top R E\,Y_{k_{\mathrm{opt}}}\bigr)
  \ \ge\ \lambda_{\min}(R)\,\lambda_{\min}(Y_{k_{\mathrm{opt}}})\,\|E\|_F^2 \ >\ 0 ,
\]
since $R\succ0$ and $Y_k\succ0$. So $\nabla^2\mathcal{L}(k_{\mathrm{opt}})\succ0$,
and Lemma~\ref{lem:nondeg} gives $loc$-P\L{}I in the sublevel form we require.
\end{proof}

Observe how the two halves of that proof come from different places. This is
the transversality of Remark~\ref{rem:transverse} in its most concrete form. The
CJS-P\L{} estimate is designed to be sharp at infinity (that is the point of
removing the restriction to compact sets) and it is deliberately not sharp at
the origin. In the scalar instance worked out in \cite{CuiJiangSontag2024} the
witness is
\[
  \xi_1(p) \ =\ \frac{2p}{\sqrt2+4p} ,
\]
which saturates at $\tfrac12$ as $p\to\infty$ but is \emph{linear} in $p$ at the
origin, whereas $loc$-P\L{}I demands order $\sqrt p$. So class $\mathcal{K}$, on
its own, does not deliver $loc$-P\L{}I; the origin has to be supplied separately,
and above it comes from the Hessian. Conversely the Hessian says nothing about
infinity. Each hierarchy contributes exactly the half the other lacks.

With Proposition~\ref{prop:lqrsgl} and Theorem~\ref{thm:cjs}(a) in hand,
Remark~\ref{rem:coercive} and Theorem~\ref{thm:one} give:

\begin{cor}\label{cor:lqr}
There is a diffeomorphism $\phi\colon D\to\R^{mn}$ with
$\mathcal{L} = \underline{\mathcal{L}} + \|\phi\|^2$. In particular $D$ is
diffeomorphic to $\R^{mn}$, and by Theorem~\ref{thm:gconvex} there is a complete
Riemannian metric on $D$ making $\mathcal{L}$ geodesically convex and globally
$1$-P\L{}.
\end{cor}

The last clause deserves a comment, because \cite{Sontag2025L4DC} remarks in
passing, next to the picture of the non-convex sublevel sets, that one ``may
convexify by coordinate change, but [it is] hard to see comparison functions for
the original problem''. Corollary~\ref{cor:lqr} says that a convexifying
coordinate change not merely exists but can be taken to make the loss an exact
convex quadratic, $\mathcal{L}\circ\phi^{-1}(y) = \underline{\mathcal{L}}+\|y\|^2$.
What it does not do is produce comparison functions, for the reason given above:
the change of variables is not an isometry, so it moves the gradient around.

\begin{rem}[Relation to the known topology of $D$]\label{rem:topology}
Two statements about the stabilizing set have now been made, and both should be
checked against what is known.

First, $D$ is connected (Proposition~\ref{prop:lqrsgl}). This is consistent with
the literature, where $D$ is known to be open, unbounded and in fact contractible
\cite{BuMesbahiMesbahi2021}; we note only that the coercivity of $\mathcal{L}$
gives a one-line proof. (The
disconnectedness phenomena in that literature concern restricted or structured
gains, output feedback, and the discrete-time setting, not unrestricted
continuous-time state feedback.)

Second, and this is the part that is not merely a repackaging,
$D\cong\R^{mn}$ \emph{as a smooth manifold}. Contractibility does not give this
by itself: that implication is exactly what fails for contractible open manifolds
in dimension three and above, which is the whole point of the Whitehead manifold
and of $\bcr$ Corollary 1.7. Upgrading ``contractible'' to ``diffeomorphic to
Euclidean space'' is what Theorem~\ref{thm:one} does, and it does it using the
LQR cost as the Morse function. We have not found the diffeomorphism statement in
the literature, but we have not looked exhaustively, and the differential-geometric
parametrizations of $D$ cited in \cite{BuMesbahiMesbahi2021} may well yield it by
another route.

One warning, since the shape of the argument invites a shortcut that does not
work. Knowing that $D\cong\R^{mn}$ does \emph{not} dispose of the completeness
hypothesis, and it is not what does so here. Completeness is a property of a
metric, not of a smooth manifold: a diffeomorphism $D\to\R^{mn}$ tells us only
that $D$ \emph{admits} a complete metric, whereas $sgl$-P\L{}I is an inequality
about $\|\nabla\mathcal{L}\|$ in the ambient Euclidean metric of $\R^{m\times n}$,
in which $D$ is not complete. Transporting along such a diffeomorphism carries the
metric with it, and one is left verifying the hypothesis in a pushed-forward metric
over which one has no control. What does the work is
Remark~\ref{rem:coercive}: coercivity replaces completeness in each of the three
places where the latter is used, and the conformal rescaling stays on $D$, changing
only the metric and weakening the witness in a controlled way. That route also
avoids a circularity, since $D\cong\R^{mn}$ is a conclusion here, not a hypothesis;
the diffeomorphism type of $D$ plays no role at all in verifying that
$\mathcal{L}$ satisfies $sgl$-P\L{}I.
\end{rem}

\subsection{Logistic regression}

Return to the two-sample example \eqref{eq:logistic}: samples $(1,0)$ and
$(2,1)$, scalar parameter, $\mathcal{L}(k)=\log(1+e^k)+\log(1+e^{-2k})$. As noted
in Section~\ref{sec:why}, $\mathcal{L}\to\infty$ at both ends while
$\mathcal{L}'\to1$ and $-2$, so there is no $\mathcal{K}_\infty$ estimate. But
$\mathcal{L}''>0$ everywhere, so $\mathcal{L}$ is strictly convex with a unique
minimizer $\bar k$ ($\bar k\approx0.4196$,
$\underline{\mathcal{L}}\approx1.2839$, $\mathcal{L}''(\bar k)\approx1.082$), and
$|\mathcal{L}'|$ increases away from $\bar k$ to the limits $1$ and $2$. Hence
$\alpha_{\mathcal{L}}$ is positive definite (increasing to
$\min\{1,2\}=1$, so in fact of class $\mathcal{K}$); and $loc$-P\L{}I holds by
Lemma~\ref{lem:nondeg}, since $\mathcal{L}$ is coercive with the single critical
point $\bar k$ and $\mathcal{L}''(\bar k)>0$. So $sgl$-P\L{}I holds and
Theorem~\ref{thm:one} applies: $\mathcal{L}=\underline{\mathcal{L}}+\phi^2$ for a
diffeomorphism $\phi\colon\R\to\R$, which here is simply
$\phi(k)=\operatorname{sign}(k-\bar k)\sqrt{\mathcal{L}(k)-\underline{\mathcal{L}}}$
and is smooth at $\bar k$ by the Morse lemma, strictly increasing by strict
convexity, and onto $\R$ by coercivity.

The general case is the same story, but the hypothesis on the data must be stated
carefully, because ``not linearly separable'' is ambiguous and the weaker readings
are not enough. Let $\{(x_i,y_i)\}_{i=1}^N$ have $x_i\in\R^n$ and
$y_i\in\{0,1\}$, write $s_i := 2y_i-1\in\{\pm1\}$, so that the empirical binary
cross-entropy is
\[
  \mathcal{L}(k) \ =\ \frac1N\sum_{i=1}^N \log\bigl(1+e^{-s_i x_i^\top k}\bigr) ,
\]
and assume
\begin{equation}\label{eq:nosep}
  \text{there is no } v\in\R^n\setminus\{0\} \text{ with }
  s_i\,x_i^\top v \ \ge\ 0 \ \text{ for every } i .
\end{equation}
This is the exact condition used in \cite{CuiJiangSontag2026}: it excludes not
only strict separation but also weak, or quasi-complete, separation, and it is what
makes $\mathcal{L}$ coercive: if a $v$ as in \eqref{eq:nosep} existed, every
summand would be nonincreasing along $t\mapsto tv$ and $\mathcal{L}$ would stay
bounded. Observe that \eqref{eq:nosep} also forces the $x_i$ to span $\R^n$, since
any $v$ orthogonal to all of them would satisfy $s_ix_i^\top v=0\ge0$; so no
separate spanning hypothesis is needed.

Under \eqref{eq:nosep}, \cite{CuiJiangSontag2026} give that $\mathcal{L}$ is
coercive, that $\nabla\mathcal{L}$ is globally Lipschitz, and that $\mathcal{L}$
satisfies the $\mathcal{K}$-P\L{} condition,
$\|\nabla\mathcal{L}(k)\|\ge\alpha(\mathcal{L}(k)-\underline{\mathcal{L}})$ for
some $\alpha\in\mathcal{K}$. Since $\mathcal{K}\subset\mathcal{PD}$, that gives
the level-wise bound; and since
\[
  \nabla^2\mathcal{L}(k) \ =\ \frac1N\sum_{i=1}^N p_i(k)\bigl(1-p_i(k)\bigr)\,x_ix_i^\top
  \ \succ\ 0
\]
(the weights are strictly positive and the $x_i$ span), $\mathcal{L}$ is strictly
convex with a unique, nondegenerate minimizer. Lemma~\ref{lem:nondeg} then gives
$loc$-P\L{}I, and Proposition~\ref{prop:sgl} gives $sgl$-P\L{}I.
Theorem~\ref{thm:one} then yields:

\begin{cor}[Cross-entropy is a nonlinear sum of squares]\label{cor:logistic}
Under the above hypotheses there is a diffeomorphism $\phi\colon\R^n\to\R^n$ with
\[
  \mathcal{L}(k) \ =\ \underline{\mathcal{L}} + \|\phi(k)\|^2
  \qquad\text{for all } k\in\R^n .
\]
That is, the logistic-regression landscape is diffeomorphic to the standard
paraboloid, by a global smooth change of parameters.
\end{cor}

Observe that of the three properties quoted from \cite{CuiJiangSontag2026}, we use
only the first and the third, and of the third only the fact that it implies the
$\mathcal{PD}$ property, together with the local nondegeneracy. The globally Lipschitz gradient, which is
what that paper needs for its noise-to-state stability conclusions, plays no role
here; conversely the nondegeneracy of $\nabla^2\mathcal{L}(\bar k)$, which is what
we need, plays no role there. This is the transversality of
Remark~\ref{rem:transverse} showing up in a single worked example.

We should also be careful not to claim too much for the corollary. It is a statement about
the smooth structure of the landscape, not about the algorithm: the change of
variables is not linear, so it does not commute with gradient descent, and the
algorithmic content remains the $\mathcal{K}$-P\L{} estimate of
\cite{CuiJiangSontag2026}. What Corollary~\ref{cor:logistic} does say, and
convexity by itself does not, is that there is a \emph{global} normal form with no
residual structure whatsoever: the loss is exactly $\|\phi\|^2$, not merely
convex, and its sublevel sets are diffeomorphic images of round balls. It is also
worth noting that the corollary uses less than convexity (only $sgl$-P\L{}I plus
a nondegenerate minimizer) and so applies to regularized or mildly nonconvex
variants, provided all of coercivity, a unique critical point, and nondegeneracy
at it survive alongside the $\mathcal{K}$-P\L{} estimate. The comparison-function
estimate alone would not suffice.

\subsection{What the structure theory adds, and what it does not}

To summarize the position of these two examples in the two hierarchies of
Remark~\ref{rem:transverse}. Both problems satisfy conditions on \emph{both}
axes: a class $\mathcal{K}$ estimate (hence small-input ISS by
\cite{CuiJiangSontag2024, Sontag2022}, and hence small-covariance noise-to-state
stability of the associated stochastic gradient dynamics by
\cite{CuiJiangSontag2026}) and $sgl$-P\L{}I (hence the normal form above). But
they satisfy them for different reasons, and Proposition~\ref{prop:lqrsgl} shows
just how different: in the LQR problem the two halves come from two different
lemmas, one sharp at infinity and one sharp at the origin. The saturating
behavior at large loss is what buys robustness; the nondegenerate minimizer plus
level-wise uniformity is what buys the normal form. Neither implies the other:
Example~\ref{ex:exp} is $\mathcal{K}_\infty$ with no minimizer, and
Example~\ref{ex:log} has a perfect normal form with $\alpha\to0$ at infinity.

It is worth stressing how systematic this is. The robustness ladder of
\cite{Sontag2025L4DC, CuiJiangSontag2026} is indexed entirely by the comparison
class, in both the deterministic and the stochastic setting:
\[
  \mathcal{K}_\infty \Rightarrow \text{ISS} , \quad
  \mathcal{K} \Rightarrow \text{siISS} , \quad
  \mathcal{PD} \Rightarrow \text{iISS} ;
  \qquad
  \mathcal{K}_\infty \Rightarrow \text{NSS} , \quad
  \mathcal{K} \Rightarrow \text{scNSS} , \quad
  \mathcal{PD} \Rightarrow \text{iNSS} .
\]
Every rung is a statement about $\alpha$ at infinity. The structure theory reads
$\alpha$ only at the origin, plus level-wise positivity. So the two bodies of
results are not competing refinements of one hypothesis; each requires a
different half of it.

So the structure theory should be read as adding a third kind of conclusion to the
list in \cite{Sontag2025L4DC}, alongside rates and robustness: the \emph{shape} of
the landscape and of its minimizer set. In the two problems above the minimizer
set is a point and the shape conclusion is a normal form. It would be more
interesting in a problem where the minimizer set is
positive-dimensional (the
overparametrized LQR formulation of \cite{deOliveiraSiamiSontag2025}, where the
critical points are the low-rank approximations of $k^*$, is the obvious
candidate), but there the critical set includes strict saddles, so no P\L{}I of
any kind holds globally, and one would have to work on the uniformly imbalanced
invariant sets where \cite{Sontag2025L4DC, WafiEtAl} recover $gl$-P\L{}I. We have
not pursued this.

\section{Discussion}
\label{sec:discussion}


The structure theory of \cite{BCR2026} is insensitive to the global size of the
gradient. What it needs is that the \L{}ojasiewicz exponent at the minimizer set
be $\tfrac12$ (this is condition \eqref{eq:A}, and Example~\ref{ex:cross}
shows it is indispensable) and that the gradient not vanish on any positive
level, uniformly along that level. Example~\ref{ex:tanh} shows that this second
requirement cannot simply be dropped either, though there the failure is narrow and
worth stating exactly: the nonlinear least-squares form and the fiber bundle
structure survive, and what is lost is the identification of the target with $\R^k$,
hence the pure quadratic of Theorem~\ref{thm:one} and the hidden convexity of
Theorem~\ref{thm:gconvex}. By Proposition~\ref{prop:sgl} those two requirements together are
exactly $sgl$-P\L{}I in the sense of \cite{Sontag2025L4DC}. Given them, the
negative gradient flow reaches $S$ in finite length, $S$ is a Morse--Bott
manifold of minimizers, the end-point map is a trivial smooth fiber bundle over
$S$ when $S$ is contractible, and $f = f^* + \|\phi\|^2$. In \KL{} language, the
hypothesis is a global inequality of \KL{} type whose desingularizer is
$O(\sqrt h)$ at the origin, ``of \KL{} type'' because we do not require the
desingularizer to be concave.

So the main conclusion is short: \emph{$gl$-P\L{}I may be replaced by $sgl$-P\L{}I
throughout \cite{BCR2026}, and neither of the two conditions that make up
$sgl$-P\L{}I can be dropped.} The conclusions are therefore available in the
continuous-time LQR and logistic regression problems that motivated the hierarchy
in the first place, both of which satisfy $sgl$-P\L{}I
(Section~\ref{sec:applications}); $sat$-P\L{}I, which sits strictly between, is
more than is needed.

Let us be careful not to claim too much. Four points seem worth stating.

First, the generalization is, in a precise sense, only apparent, though the
precise sense matters, and Remark~\ref{rem:collapse} is where it is stated. By
Theorem~\ref{thm:gconvex}, on a contractible manifold any $f$ in our class
becomes globally $1$-P\L{} after a change of complete metric, the new metric
depending on $f$; and by Corollary~\ref{cor:charS} the family of realizable
pairs $(\M,S)$ is unchanged. Relative to a \emph{fixed} metric the class is
genuinely larger, by Examples~\ref{ex:bounded} and \ref{ex:log}; it is only
after quantifying over metrics that the two classes agree. So the enlargement
adds functions but no geometry. One could regard
this as a defect of the result or as the point of it, depending on taste; we
lean towards the latter, since a hypothesis that can be normalized away is a
hypothesis one should not have been assuming.

Second, the cost is quantitative: global quadratic growth is genuinely lost, and
Example~\ref{ex:log} shows how badly. Anyone wanting a global estimate must
reintroduce control on $\alpha$ away from $S$, at which point one is back to
something like (J3).

Third, we have not touched the two hypotheses that \cite{BCR2026} itself flags
as the interesting ones: contractibility of $\M$, where the four examples of
$\bcr$ Section 7 remain in force verbatim and continue to obstruct the obvious
weakenings; and $C^\infty$ regularity, which enters through
\cite{RB2024a} in Lemma~\ref{lem:MB} and could presumably be relaxed to $C^p$
with $p$ large, at the price of the usual bookkeeping.

Fourth, and this is the point we would most want a reader from control theory to
take away: the normal form says nothing about robustness, and the
comparison-function class says nothing about the normal form
(Remark~\ref{rem:transverse}). The two live at opposite ends of the same
inequality. It is tempting, having proved that $\mathcal{L}$ is
$\underline{\mathcal{L}}+\|\phi\|^2$, to expect ISS-like conclusions to follow;
they do not, because $\phi$ is not an isometry and the perturbed flow is not the
perturbed flow of $\|\phi\|^2$. The robustness theory of \cite{Sontag2025L4DC,
CuiJiangSontag2024, CuiJiangSontag2026} is not superseded by anything here.

Four questions we would like to see answered.

\begin{enumerate}
\item Exactly how much uniformity is needed? Example~\ref{ex:tanh} shows that the
compact form of $sgl$-P\L{}I does not suffice, and that the sublevel form is the
right one. What remains open is the narrow gap described at the end of
Remark~\ref{rem:invex}: an $f$ with $h(\M)=[0,\infty)$, that is, attaining its
infimum and unbounded above, whose sharpest $\alpha_f$ is positive on every
positive level but is not bounded below on some band of levels. There
no continuous positive definite witness exists, so $\Psi$ is unavailable and our
proof fails; we know of no counterexample to the conclusions.
\item Is the exponent $\tfrac12$ needed, or only the nondegeneracy it implies?
Example~\ref{ex:cross} rules out exponent $\tfrac34$, but the natural
formulation of the question is in terms of a Morse--Bott hypothesis imposed
directly, as in the first bullet of $\bcr$ Section 8.
\item Given $\alpha$, what quantitative control on the diffeomorphism $\psi$ of
Theorem~\ref{thm:main} can be asserted away from $S$? Near $S$ the singular
values of $D\psi$ are pinned by \eqref{eq:MB}, and a clean statement in terms of
$\Psi$ ought to exist further out.
\item Is there an application with a positive-dimensional minimizer set? The
conclusions of \cite{BCR2026} are most interesting when $\dim S>0$, and both
applications in Section~\ref{sec:applications} have $S$ a single point, where the
content reduces to a global Morse lemma. The overparametrized formulations of
\cite{deOliveiraSiamiSontag2025, WafiEtAl} do have positive-dimensional critical
sets, but they also have strict saddles, so no P\L{}I holds globally there; the
question is whether the invariant sets on which $gl$-P\L{}I is recovered are large
enough, and shaped well enough, for the fiber bundle picture to say something.
\end{enumerate}

\subsection*{Acknowledgment}

The structural results reformulated here are entirely those of Boumal,
Criscitiello and Rebjock \cite{BCR2026}; our contribution is only to notice which
hypothesis their proofs use, and where. The comparison-function hierarchy in which
that hypothesis turns out to sit, and the two motivating problems of
Sections~\ref{sec:why} and~\ref{sec:applications}, are from
\cite{Sontag2025L4DC} and the work summarized there, in particular
\cite{CuiJiangSontag2024, CuiJiangSontag2026}. The Morse--Bott input is due to
Rebjock and Boumal \cite{RB2024a}, and the smoothness of the end-point map to
Falconer \cite{Falconer1983}.

It is a pleasure to thank Nicolas Boumal for a careful reading of an earlier
version of this note. Two of his suggestions are incorporated directly: the
polynomial reparametrization $\theta(t)=t(1+t)/\mu$ in Example~\ref{ex:bounded},
and the single formula for $\Lambda$ in the proof of
Proposition~\ref{prop:shortcut}, which replaces a clumsier two-step construction.
His observation that his $\theta$ serves the same purpose as $\tfrac12\Psi^2$ while
differing from it is what led us to the identity \eqref{eq:average} and to the
characterization of $\tfrac12\Psi^2$ as a pointwise lower bound for the admissible
reparametrizations, and as the minimal admissible one whenever it is smooth
(Remark~\ref{rem:minimal}), and thence to the fact that the obstructions exhibited
in Remarks~\ref{rem:shrink} and~\ref{rem:sqrtorder} are attached to the witness
rather than to $f$. He also proposed a reformulation of the whole argument in which
the existence of the reparametrization is taken as the hypothesis and the work of
this note is to supply it; that vantage point is what suggested
Remark~\ref{rem:equiv}, and, because his formulation allows the image of $\Lambda$
to be a ball of finite radius, the closer reading of Example~\ref{ex:tanh}.

\subsection*{A note on the use of AI}

An extended dialogue with several AI systems assisted the author in drafting and polishing the exposition, carrying out computational and logical checks, suggesting ideas, and locating references.


\end{document}